\documentclass[11pt]{scrartcl}
\usepackage[utf8]{inputenc}
\usepackage{cmbright}
\usepackage[usenames,dvipsnames]{xcolor}
\usepackage{graphicx,longtable,url,rotating}
\usepackage{float}
\usepackage{wrapfig}
\usepackage{amsmath}
\usepackage{amsfonts}
\usepackage{amssymb}
\usepackage{stmaryrd}
\usepackage{subfig}
\usepackage{algpseudocode}
\usepackage{alltt}
\usepackage{tikz}
\usepackage{mathtools}

\definecolor{lightcolor}{gray}{.55}
\definecolor{shadecolor}{gray}{.95}
\usepackage{bbm}
\usepackage{algorithm}
\usepackage{adjustbox}
\usepackage{changes}
\usepackage{truncate}
\definechangesauthor[name={Christophe}, color=orange]{xtof}
\usepackage{hyperref}
\hypersetup{colorlinks=true,pagebackref=true,urlcolor=orange,citecolor=brown}
\author{Morgan André and Christophe Pouzat}

\DeclareMathAlphabet{\pazocal}{OMS}{zplm}{m}{n}
\DeclareRobustCommand{\rchi}{{\mathpalette\irchi\relax}}
\newcommand{\irchi}[2]{\raisebox{\depth}{$#1\chi$}}

\usetikzlibrary{calc, positioning, automata, chains, backgrounds, arrows.meta}
\tikzstyle{vertex}=[circle, draw, inner sep=0pt, fill=black, minimum size=8pt]
\tikzstyle{sqvertex}=[diamond, draw, inner sep=0pt, fill=black, minimum size=9pt]
\usetikzlibrary{shapes}

\makeatletter
\newcommand*\bigcdot{\mathpalette\bigcdot@{.5}}
\newcommand*\bigcdot@[2]{\mathbin{\vcenter{\hbox{\scalebox{#2}{$\m@th#1\bullet$}}}}}
\makeatother

\newtheorem{theorem}{\indent Theorem}[section]
\newtheorem{proposition}[theorem]{\indent Proposition}

\newtheorem{lemma}[theorem]{\indent Lemma}

\newtheorem{definition-theorem}[theorem]{\indent Definition-Theorem}

\newenvironment{proof}{\paragraph{Proof:}}{\hfill$\square$}

\def \R{\mathbb{R}}

\def \P{\mathbb{P}}

\title{A Quasi-Stationary Approach to Metastability in a System of Spiking Neurons with Synaptic Plasticity}

\author{M. André$^1$ \and C. Pouzat$^2$}
\hypersetup{
 pdfauthor={Morgan André and Christophe Pouzat},
 pdftitle={Metastability in a System of Spiking Neurons with Synaptic Plasticity},
 pdfkeywords={},
 pdfsubject={},
 pdfcreator={}, 
 pdflang={English}}

\begin{document}
\maketitle

\section*{Affiliations}

$^{1}$ Instituto de Matemática e Estatística - USP, University of São Paulo \\
$^{2}$ Strasbourg University and CNRS UMR 7501\\

\noindent\textbf{Correspondance should be addressed to}:\\
C. Pouzat\\
IRMA\\
7 rue René-Descartes\\
67084 Strasbourg Cedex\\
France\\
christophe.pouzat@math.unistra.fr\\

\section*{Abstract}
After reviewing the behavioral studies of working memory and of the cellular substrate of the latter, we argue that metastable states constitute candidates for the type of transient information storage required by working memory. We then present a simple neural network model made of stochastic units whose synapses exhibit short-term facilitation. The Markov process dynamics of this model was specifically designed to be analytically tractable, simple to simulate numerically and to exhibit a quasi-stationary distribution (QSD). Since the state space is finite this QSD is also a Yaglom limit, which allows us to bridge the gap between quasi-stationarity and metastability by considering the relative orders of magnitude of the relaxation and absorption times. We present first analytical results: characterization of the absorbing region of the Markov process, irreducibility outside this absorbing region and consequently existence and uniqueness of a QSD. We then apply Perron-Frobenius spectral analysis to  obtain any specific QSD, and design an approximate method for the first moments of this QSD when the exact method is intractable. Finally we use these methods to study the relaxation time toward the QSD and establish numerically the memorylessness of the time of extinction.

\pagebreak
\section*{Acknowledgments}
This article was produced as part of the activities of FAPESP Research,
Innovation and Dissemination Center for Neuromathematics (grant number
2013/07699-0, S.Paulo Research Foundation, Brazil), and Morgan André was
supported by FAPESP, Brazil scholarships (grant number  2020/12708-1). Christophe Pouzat was supported by the ANR project Simbadnesticost (ANR-22-CE45-0027-01). Eva Löcherbach saw that we were dealing with a quasi-stationary distribution and put us on the right track. We thank our two anonymous referees for gently pushing us to refine our analytical and numerical arguments. 
Finally and very sadly, this article is dedicated to the memory of our friend and colleague Antonio Galves. Antonio was in many ways---through his contributions on metastability and his constant drive for simple (but not too simple) models amenable to analytical treatment---at the origin of the present work.
\pagebreak
\section{Introduction}
\subsection{The neurobiological horizon}

As described by Fuster in the introduction of his 1973 article \cite{fuster:73}:
\emph{A delayed-response trial typically consists of the presentation of one of two possible visual cues, an ensuing period of enforced delay and, at the end of it, a choice of motor response in accord with the cue. The temporal separation between cue and response is the principal element making the delayed response procedure a test of an operationally defined short-term memory function.}
In that article Fuster described, in the monkey prefrontal cortex, neurons that switch between no activity and a \emph{sustained activity at constant rate during the delay period} when the animal had to perform a delayed-response task. He showed moreover---using distracting stimuli that interrupted the sustained activity of these neurons---that the monkey errors at the end of the delay period were positively correlated with the interruption of sustained activity. Since then, many experimental investigations reviewed in \cite[Chap. VII]{fuster:15} and \cite{constantinidis.ea:18} have confirmed this basic finding and showed that some of the ``sustained activity neurons'' are insensitive to the type of cue (color, shape, location, sound) and seem to encode the ``abstract'' notion of remembering ``any'' cue until the expiration of a delay, while others, especially outside the prefrontal cortex, are sensitive to the type of cue. These ``sustained activity neurons'' are relatively easy to record from, \emph{implying that they are fairly abundant} \cite{leavitt.ea:17}. This sustained activity has been intriguing modelers for a long time, leading them to explore first network models with subgroups of strongly reciprocally coupled excitatory neurons   \cite{zipser.kehoe.ea:93,amit.brunel:97}. The sustained activity has then been
interpreted as a local attractor of some dynamical system (reviewed in \cite{wang:01}). Stability issues when the transiently memorized item is a continuous quantity--like an angle--, led modelers to include some ``slow'' and ``use dependent'' coupling, initially in the form of NMDA receptors \cite{compte.ea:00,wang:01}--for a review of basic neurophysiology, see \cite{luo:15}. But \cite{wang.ea:06} described a subclass of pyramidal (and therefore excitatory) cells in the prefrontal cortex that are strongly interconnected and whose synapses are unusual, since they exhibit a marked short-term facilitation--synapses between neurons of this type usually exhibit short-term depression. This has lead to several studies giving a more or less central role to short-term facilitation in sustained activity generation or stabilization, \emph{e.g.} \cite{barak.tsodyks:07,itskov.ea:11,hansel.mato:13}--or even proposing a working memory mechanism without sustained activity \cite{mongillo.ea:08}--, reviewed in \cite{barak.tsodyks:14}. But the secondary status of the ``noise'' in these studies, where variability comes into play mostly at the neurons input level, is at odd with basic empirical observations. It is indeed well known \cite{defelice:81,Yarom_2011,luo:15} that neurons depend on ion channels that are randomly going back and forth between closed and opened states both for the action potential generation \cite{verveen.derkesen:68} and the synaptic transmission \cite{katz.miledi:70}; that (chemical) synaptic transmission involves the release of a variable number of transmitter packets / quanta \cite{Fatt_1952,del_Castillo_1954}, giving rise to the rather noisy membrane voltage trajectories that are actually observed. These considerations strongly suggest an alternative model construction strategy: working with stochastic units / neurons instead of deterministic ones. Continuing and simplifying \cite{galves.locherbach.ea:19}, we therefore developed a minimal model of the sub-network of reciprocally coupled pyramidal cells with facilitating synapses \cite{wang.ea:06}; a model that is both amenable to analytical solutions and that can be easily simulated. This model is made of stochastic neurons that accumulate their inputs until a threshold is reached. The synapses between the neuron exhibit short-term facilitation enabling the sub-network to exhibit a transient \emph{memoryless} sustained activity--that is, to reach a quasi-stationary distribution leading to genuine metastability--, reminiscent of what is observed in working memory experiments.

\subsection{Metastability and Quasi-stationarity}

The expressions ``metastability" and ``quasi-stationary distribution" used in the previous section have proper mathematical meanings, which deserve explanations. The notion of metastability on one hand has a long history, and appears in a wide variety of fields to describe various apparently unrelated phenomenons, from nuclear physics to the study of avalanches and super-cooling water. More recently it has gained popularity in the field of neuroscience, as in many respects the brain seems to exhibit metastable-like behaviors. In the present article the notion of metastability is to be understood as in the rigorous characterization introduced in \cite{cassandro} for interacting particle systems, which requires that, for a Markov process evolving in a space which includes an absorbing state: \textit{(i)} the time to reach the absorbing state from any other state is memoryless (i.e. follows an exponential distribution), asymptotically with respect to the number of particles in the system; \textit{(ii)} before reaching this absorbing state (the only real equilibrium), the system exhibits \textit{thermalization}, i.e. an apparent stabilization, temporary but long, in a region of the state space away from the actual equilibrium. The notion of quasi-stationary distribution (QSD) on the other hand refers to the stationary distribution of a modified Markov process which has been conditioned to stay away from its absorbing state (see \cite{iosifescu:07,meleard:2012} for an introduction). Since here we consider a system which is essentially irreducible and evolves in a finite state space, the QSD is unique and corresponds to the Yaglom limit, that is, the unique limit distribution to which the system conditioned on non-absorption relaxes, starting from any given state of the irreducibility region. Moreover it is well-known that, in this case, if the initial state is distributed with respect to the QSD then the time of extinction is indeed memoryless. In the sequel we argue that the thermalization referred to above (that is, the metastable phase) can be understood through the theory of quasi-stationarity --- a possibility which, quite surprisingly, seems to have been mostly ignored in the literature, with the notable exceptions of \cite{bianchi:2016} and \cite{huisinga:2004}, in specific settings which are not applicable in our case. This approach gives us analytical tools to quantitatively characterize the metastable phase, as well as a pretty straightforward (but not yet rigorous) way of establishing the memorylessness evoked in point \textit{(i)}.

\subsection{Overview}

The present paper is organized as follows. In Section \ref{sec:definition} we give the basic definition of our model: first from the modeling viewpoint; then in a somewhat more mathematical form, as a continuous-time Markov chain. In Section \ref{sec:cutting} we partition the state space into relevant sub-regions; in particular we give an explicit characterization of the absorbing region, and of the support of the (yet to be proven) metastable phase; irreducibility is also established. In Section \ref{sec:QSD}, we introduce briefly the relevant elements of the theory of quasi-stationarity distributions; the key role of the Perron-Frobenius spectral analysis is emphasized. In Section \ref{sec:nsmall}, these ideas are implemented for a simple example with $5$ neurons, which allows us to compare numerically the relative orders of magnitude of the relaxation toward the QSD and then toward the absorbing region, establishing the memorylessness of the time of extinction. Since this exact method becomes rapidly intractable, we then study numerically the extinction time in Section \ref{sec:nbig}. An heuristic calculation that provides an approximate alternative to the spectral approach is presented in Section \ref{sec:approx}.

\section{Definition of the model}
\label{sec:definition}

We start by defining informally a stochastic system of interacting spiking neurons --- a formal definition will be given in the next section. The system consists in a finite set of \(N\) neurons. A \emph{membrane potential process}, denoted \((U_t(i))_{t \geq 0}\), taking values in the set of non-negative integers is associated to every neuron \(i \in \{1,\ldots N\}\). The spiking activity of the neurons depends on a threshold value \(\theta \in \mathbb{Z}^+\). When \(U_t(i) < \theta\) neuron \(i\) cannot spike, and we say that it is \emph{quiescent}, while when \(U_t(i) \geq \theta\) we say that neuron \(i\) is \emph{active}: it spikes at rate \(\beta\) --- i.e. it waits a time \(\Delta_t\) distributed as an exponential random variable of parameter \(\beta\) and spikes at time \(t + \Delta_t\). Every neuron is connected to every other neuron of the network with a uniform \emph{synaptic strength} and the effect of a spike depends on the state of the synapse between the spiking neuron and the other neurons of the network at that time. At any time the synapse can be either \emph{facilitated} or \emph{not facilitated}, meaning that if a spike occurs, it will or won't be "transmitted" to the other neurons. The facilitation state of the synapse of neuron \(i\) is denoted \((F_t(i))_{t \geq 0}\); it is a stochastic process taking value in \(\{0,1\}\). Whenever \(F_t(i) = 1\) we say that the synapse of neuron \(i\) is \emph{facilitated} at time \(t\). The facilitated synapse looses its facilitation  (i.e. goes back to the unfacilitated state \(F_{t+\Delta_t}(i) = 0\)) at a given rate \(\lambda\). Now the fact that a given neuron \(j\) spikes at time \(t\) means the following: \textit{(i)} its membrane potential is \emph{reset to} \(0\), \(U_{t^+}(j) = 0\); \textit{(ii)} its synapse becomes facilitated if it wasn't already, \(F_{t^+}(j)=1\); \textit{(iii)} if its synapse was facilitated at the time of the spike, the spike is said to be \emph{efficient} and the membrane potential process of \emph{all the other neurons} in the system increases by one unit, while nothing happens to these neurons if the synapse was not facilitated at the moment of the spike,  \(U_{t^+}(i) = U_{t^-}(i) + F_{t^-}(j)\), for all \(i \neq j\); in the latter case the spike is said to be \emph{inefficient}. All exponential random variables involved are assumed to be independent. Then a trajectory of the system is entirely characterized by the family of interacting stochastic processes $(U_t(i),F_t(i))_{t \geq 0}$ (for \(i \in \{1,\ldots N\}\)).

\subsection{Formal definition as a continuous-time Markov chain}
\label{sec:formal-definition-as-a-continuous-time-markov-chain}
It is clear that although the membrane potential any neuron $i$ can get arbitrarily large in the above model formulation, from the dynamics viewpoint all that matters is to know whether \(U_t(i) < \theta\) or not. We can therefore consider that there are only \(\theta+1\) relevant membrane potential states: \(U_t(i) = 0\), \(U_t(i) = 1\), \(\ldots\), \(U_t(i) = \theta-1\) and \(U_t(i) \ge \theta\). The effective state space of a neuron if therefore finite with \(2 (\theta+1)\) elements: \[\Phi \equiv \{0,\ldots,\theta\} \times \{0,1\}\,, \] where the set \(\{0,1\}\) corresponds to the synaptic facilitation. Figure \ref{fig:single-neuron-chain} illustrates the effective states accessible to an arbitrary neuron \(i \in \{1,\ldots,N\}\) of the network as well as as the possible transitions among those states. The state of any neuron of the network can therefore be fully specified by placing a token on one of the spaces (circles/squares, by analogy to the board of the game of the goose) of Fig. \ref{fig:single-neuron-chain} and the state of the whole network can be represented by placing N tokens labelled 1 to N on the \(2 (\theta+1)\) spaces. The state of the network is then specified by selecting a single element of the set: \(\rchi = \Phi^N\). We adopt the following notation: for any $x \in \rchi$ we write $x=(x_1, \ldots x_N)$, with $x_i = (x_i^U,x_i^F)$ for any $i \in \{1, \ldots N\}$, where $x_i^U \in \{0, \ldots \theta\}$ corresponds to the value of the membrane potential for neuron $i$ while $x_i^F$ corresponds to the facilitation state. Notice that the state space size increases very quickly with $N$ and $\theta$:
\begin{itemize}
\item for $N=5$ and $\theta=1$, the size is, $4^5=1024$,
\item for $N=50$ and $\theta=10$, it becomes roughly $2.2^{43}\, \mathcal{N}_A$, where $\mathcal{N}_A$ is Avogadro's constant.
\end{itemize}

\begin{center}
	\begin{figure}[H]
		\centering
		\includegraphics[width=0.5\textwidth]{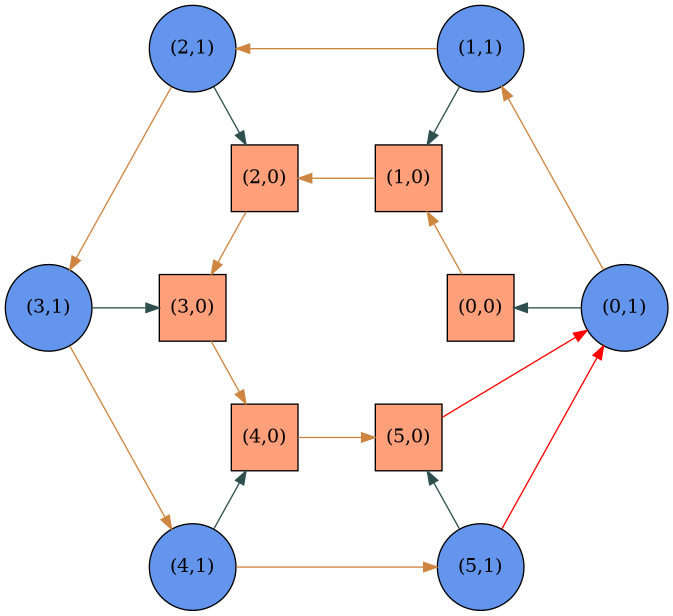}
		\caption{\label{fig:single-neuron-chain}An illustration where \(\theta = 5\). Each node represents a possible membrane potential value, the first element of the pair (being understood that "5" should be interpreted as "\(\ge 5\)") and a synaptic facilitation value in the second element of the pair. If the neuron has it synapse facilitated it sits in the outer circle (light blue circle), if its synapse is not facilitated it sits in the inner circle (light orange squares). The transition rates are encoded by the color: brown is the rate of "effective spikes" generated by the network (spike occurring while the synapse of the spiking neuron is still facilitated), dark green is \(\lambda\) and red is \(\beta\).}
	\end{figure}
\end{center}

Formally we're considering a continuous time Markov process $(X_t)_{t \geq 0}$ taking value in the finite state space $\rchi$. Its dynamic can be characterized by an infinitesimal generator $Q = (q_{x,y})_{x,y \in \rchi}$, which is defined by the following requisites. Let $x,y \in \rchi$ such that $x \neq y$ and let $i \in \{1, \ldots N\}$, then\footnote{Here and in the sequel of this article $\delta$ denotes the Kronecker delta function.}:

\begin{itemize}
	\item \textbf{Loss of facilitation :} Suppose $y$ is such that $y_j = x_j$ for any $j \neq i$ and $y_i = (x^U_i,0)$ then $q_{x,y} = \lambda \delta_1(x_i^F)$.
	\item \textbf{Inefficient spike :} Suppose $y$ is such that $y_j = x_j$ for any $j \neq i$ and $y_i = (0,1)$, then $q_{x,y} = \beta \delta_{(\theta,0)}(x_i)$.
	\item \textbf{Efficient spike :} Suppose $y$ is such that $y_j = \left(\min(\theta,x^U_j + 1), x^F_j\right)$ for any $j \neq i$ and $y_i = (0,1)$, then $q_{x,y} = \beta \delta_{(\theta,1)}(x_i)$.
\end{itemize}

Moreover for any $y \neq x$ not considered above $q_{x,y} = 0$ and of course $q_{x,x} = - \sum_{y \neq x} q_{x,y}$.
The key features of this model dynamics are illustrated next with simulations.

\subsection{Basic dynamics features}
\label{sec:basic-dynamics-features}
\subsubsection{Process trajectories}
\label{sec:first-process-trajectoriers}
Simulation details are provided in Se.~\ref{sec:simulations}. Figure \ref{fig:MPP1} shows the trajectories of the ``membrane potentials'', $u_t(i)$ a realization of $U_t(i)$, and the synaptic facilitation, $f_t(i)$ a realization of $F_t(i)$, of all the (50) neurons of a simulation ($N=50$, $\theta=5$, $\beta=10$, $\lambda=6.7$) during one time unit. 
\begin{figure}[H]
  \begin{center}
    \includegraphics[width=.7\textwidth]{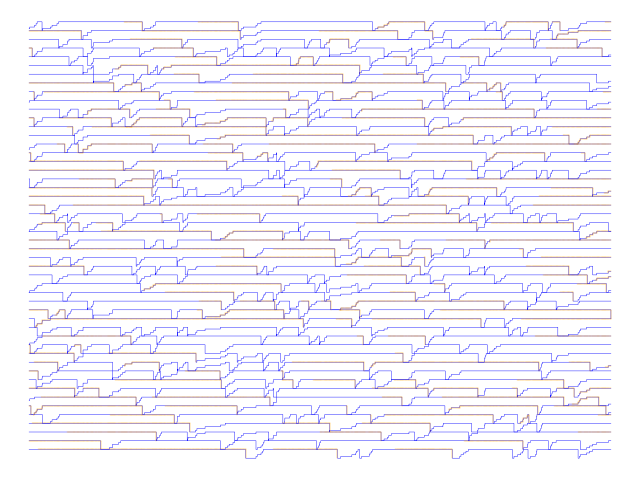}
  \end{center}
  \caption{\label{fig:MPP1}Trajectories between time units 1 and 2 of the membrane potentials of the fifty neurons of a simulated network. The traces are blue when the synapse of the neuron is facilitated and orange otherwise.}            
\end{figure}

\emph{This figure displays the complete state of the network}. Notice that at any given time, most of the neurons are in the susceptible state (their membrane potential is $\ge \theta$). Notice also that the membrane potentials of the neurons that have not yet reached $\theta$ evolve in parallel. Spike are emitted when the membrane potential of one neuron goes from $\theta$ to zero, this is the only way the membrane potential can decrease.

A finer time display is proposed on Fig.~\ref{fig:MPP1zoom}. 

\begin{figure}[H]
  \begin{center}
    \includegraphics[width=.6\linewidth]{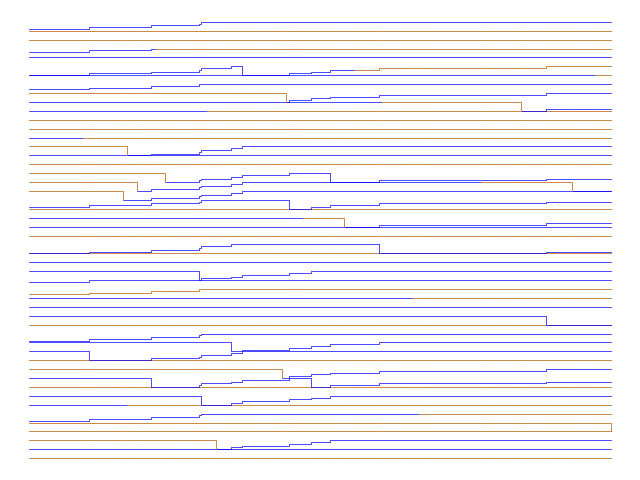}
  \end{center}
  \caption{\label{fig:MPP1zoom}Enlarge display between times 1.20 and 1.25 of the data shown of Fig.~\ref{fig:MPP1}.}
\end{figure}
The features of the model are clearly visible:
\begin{itemize}
\item When a neuron with an un-facilitated synapse spikes (the trace is \textcolor{orange}{orange} when the membrane potential is at the threshold level just before dropping to 0):
  \begin{itemize}
  \item its synapse gets facilitated (the trace turns \textcolor{blue}{blue}) immediately after the spike,
    \item the membrane potential of all the other neurons remains the same.
    \end{itemize} 
  \item When a neuron with a facilitated synapse spikes (the trace is \textcolor{blue}{blue} when the membrane potential is at the threshold level just before dropping to 0):
  \begin{itemize}
  \item its synapse remains facilitated (the trace stays \textcolor{blue}{blue}) immediately after the spike,
    \item the membrane potential of \emph{all} the other neurons that are below threshold increases by 1.
    \end{itemize} 
\end{itemize}      

\subsubsection{It is the same but it is not the same}

We now turn to the key property our model was designed to exhibit. The next two figures (\ref{fig:raster-simA} and \ref{fig:raster-simB}) show spike trains displayed as raster plots (every spike is represented by a dot) of the same network of 50 neurons started from the same initial state but using two different sequences of (pseudo) random numbers. Fig.~\ref{fig:raster-simA} shows an abrupt disappearance of the activity (after 13 time units). 
\begin{figure}[H]
  \begin{center}
    \includegraphics[width=.6\textwidth]{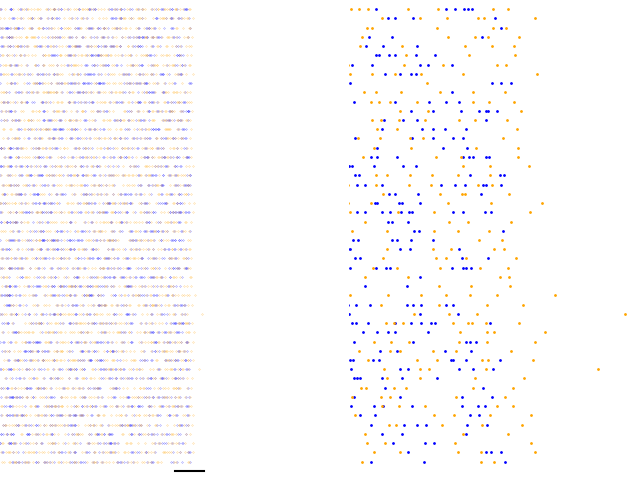}
  \end{center}
  \caption{\label{fig:raster-simA}Raster plots of a $50$ neurons network, with $\lambda = 6.7$, $\beta = 10$ and $\theta = 5$. Left, from time 0 to 14; right from time 12 to 14. Dots are blue when the synapse is facilitated and orange otherwise}
\end{figure}

The dots color is blue when the synapse is facilitated and orange otherwise. We see on the right side of Fig.~\ref{fig:raster-simA} that the last spikes occurring before the ``network death'' are all with an un-facilitated synapse. Fig.~\ref{fig:raster-simB} shows the same network as Fig.~\ref{fig:raster-simA}, starting from the same state and remaining active for whole simulation (50 time units).

\begin{figure}[H]
  \begin{center}
    \includegraphics[width=.58\textwidth]{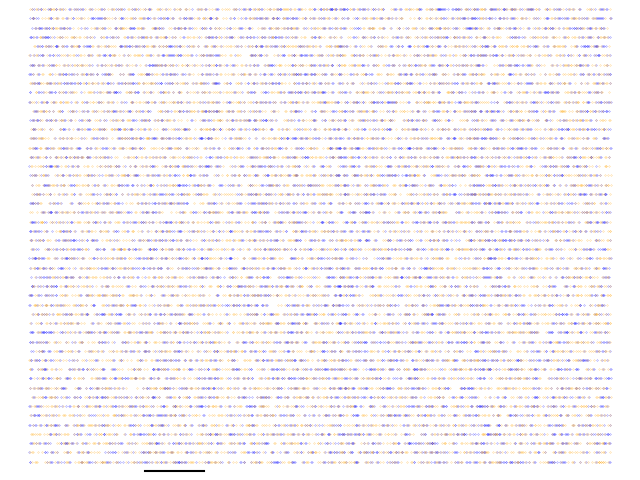}
  \end{center}
  \caption{\label{fig:raster-simB}Same as Fig.~\ref{fig:raster-simA} but different random numbers sequence. The scale bar is drawn between time 10 and time 15.}
\end{figure}

Judging from the dots pattern, the activity looks regular with a constant ratio of blue dots over orange dots. But a better way to graphically asses the network activity (network spiking frequency) is provided by the observed counting process (a step function that increases by one every time an event occurs) as shown on Fig.~\ref{fig:simA-simB-early-CP-plots} for the two simulations of Fig.~\ref{fig:raster-simA} and \ref{fig:raster-simB}.
\begin{figure}[H]
  \begin{center}
    \includegraphics[width=.7\textwidth]{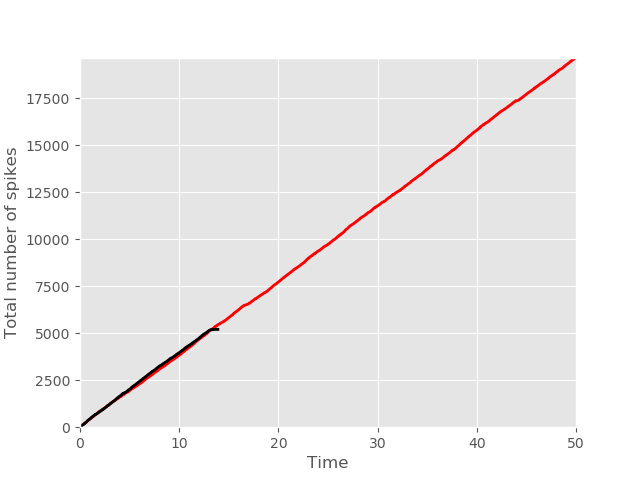}
  \end{center}
  \caption{\label{fig:simA-simB-early-CP-plots}Observed counting processes for the simulations of Fig.~\ref{fig:raster-simA} (black) and Fig.~\ref{fig:raster-simB} (red).} 
\end{figure}
Extracting the slope by eye, we see that the network generates roughly 375 events per time unit (before it reaches the quiescent state in the case of the black trace). \emph{We have, qualitatively at least, the behavior we are interested in: the activity seems ``stationary'' until it abruptly vanishes}.

\subsection{The aggregated process}
\label{sec:aggregated-process}
Notice now that the dynamic of our system is invariant with respect to the permutation of the neurons: they are all equivalent members of the network. We are moreover interested in the network state, as opposed to the individual neuron states. For these reasons we might focus on the number of neurons in each possible state at each time. In order to do this, for any $x \in \rchi$ and any $(i,j) \in \{0, \ldots \theta\} \times \{0,1\}$ we define $$z^{(i,j)}(x) = \sum_{k=1}^N \delta_i(x^U_k)\delta_j(x^F_k).$$

Then $z^{(i,j)}$ is simply a function from $\rchi$ to $\{0, \ldots N\}$ counting the number of neurons having membrane potential $i$ and facilitation state $j$. For any $(i,j) \in \{0, \ldots \theta\} \times \{0,1\}$ and any time $t \geq 0$ we might now define the stochastic version of these counting functions: $$ Z_{t}(i,j) = z^{(i,j)} (X_t).$$

Furthermore, for any $t \geq 0$ we write $Z_t = \left(Z_{t}(i,j)\right)_{(i,j) \in \{0, \ldots \theta\} \times \{0,1\}}$, and the resulting process $(Z_t)_{t \geq 0}$ is a continuous-time Markov chain taking value in $S = \{0, \ldots, N\}^{\{0, \ldots, \theta\} \times \{0,1\}}$. For any $z \in S$ we adopt the following notation: $z_{i,j}$ denotes the value in $\{0, \ldots, N\}$ corresponding to the coordinates $i \in \{0, \ldots \theta\}$ and $j \in \{0,1\}$, while $z_i$ denotes the ordered pair $(z_{i,0},z_{i,1})$. Moreover for any $i \in \{0, \ldots \theta\}$ we write $$z_{i, \bullet} = z_{i,0} + z_{i,1}.$$ 

In the sequel, if a state $z$ is such that $z_{i,\bullet} = k$ we adopt the terminology of saying that there are $k$ neurons at \emph{level} $i$. Furthermore notice that because the number of neurons in the system is fixed to $N$ we may narrow down the state space a bit; the range of the process (the actual state space) is $R \subset S$, defined as: $$R = \left\{z \in S: \sum_{i=0}^\theta z_{i,\bullet} = N\right\}.$$
The size of this state space is much smaller than the one we started from, since instead of having $\left(2\,(\theta + 1)\right)^N$ it has ``only'' \cite[Sec. II.5, p. 38]{feller:68}: $\binom{N + 2 \theta + 1}{2 \theta + 1}$  elements. Compared to the previous considered cases we get:
\begin{itemize}
\item for $N=5$ and $\theta=1$, the size is, $56$,
\item for $N=50$ and $\theta=10$, it becomes roughly $5.5 \times 10^{17}$.
\end{itemize}

It is easy to obtain the dynamic of the process $(Z_t)_{t \geq 0}$ from the definition of $(X_t)_{t \geq 0}$. We write this dynamic explicitly in term of maps from $R$ to $R$ corresponding to the three possible events susceptible to affect the system. The map corresponding to efficient and inefficient spikes will be denoted respectively $\pi$ and $\pi^*$, and the map corresponding to a loss of facilitation on a neuron at level $i$ will be denoted $\pi^\dagger_i$. Suppose the current state is $z \in R$.

\begin{itemize}
	\item \textbf{Loss of facilitation:} When it happens at level $i \in \{0, \ldots \theta\}$ the number of un-facilitated neurons at level $i$ increases by $1$, and the number of facilitated neurons decreases by $1$. The map $\pi^\dagger_i$ is therefore defined by:
	
	$$\left(\pi^\dagger_i(z)\right)_j = \left\{ 
	\begin{array}{ll}
		(z_{i,0} + 1,z_{i,1} - 1) \text{ if } j=i,\\
		(z_{j,0},z_{j,1}) \text{ otherwise.}\\
	\end{array}
	\right. $$ 
	
	\item \textbf{Inefficient spike:} This event leads to a decrease of the non-facilitated neurons at level $\theta$ by one and an increase of the number of facilitated neurons at level $0$ by one (red arrow from square $(5,0)$ to circle $(0,1)$ in Fig. \ref{fig:single-neuron-chain}). The map $\pi^*$ is therefore defined by: 
	
	$$\Big(\pi^*(z)\Big)_j = \left\{ 
	\begin{array}{ll}
		(z_{\theta,0} - 1, \ z_{\theta,1}) \text{ if } j=\theta,\\
		z_j \text{ if } j \in \{1, \ldots \theta-1\},\\
		(z_{\theta,0}, \ z_{\theta,1} + 1) \text{ if } j=0.\\
	\end{array}
	\right.$$ 
	
	\item \textbf{Efficient spike:} This event leads to a decrease of the number of facilitated neurons at level $\theta$ by one \emph{and} an increase by $z_{\theta-1,1}$ (when the spike comes from a neuron with a facilitated synapse, all neurons get their membrane potential increased by one except the one that spiked, whose membrane potential drops to 0), while the number of non-facilitated neurons at level $\theta$ increases by $z_{\theta-1,0}$. The number of facilitated neurons at level $0$ is set to $1$, and the number of non-facilitated neurons at level $0$ is set to $0$. Moreover for $i \in \{1, \ldots, \theta-1\}$, the number of facilitated (resp. non-facilitated) neurons at level $i$ is set to $z_{i - 1,1}$ (resp. $z_{i - 1,0}$). On Fig. \ref{fig:single-neuron-chain}, the contents of all the circles and squares rotate by one step counter clockwise, except for element \((0,0)\) that becomes 0, element \((0,1)\) that becomes 1 and elements \((5,0)\) that adds the content of element \((4,0)\) to its own and element \((5,1)\) that also adds the content of \((4,1)\) to its own and decrease by 1. The map $\pi$ is defined by:

	$$\Big(\pi(z)\Big)_j = \left\{ 
	\begin{array}{ll}
		(z_{\theta,0} + z_{\theta-1,0}, \ z_{\theta,1} - 1 + z_{\theta-1,1}) \text{ if } j=\theta,\\
		(z_{j-1,0}, \ z_{j-1,1}) \text{ if } j \in \{1, \ldots \theta-1\},\\
		(0, 1) \text{ if } j=0.\\
	\end{array}
	\right.$$ 
\end{itemize}

Then the infinitesimal dynamic of the process $(Z_t)_{t \geq 0}$ is given by the following: starting from some state $z \in R$, inefficient spikes occur at rate $\beta z_{\theta,0}$, efficient spikes occurs at rate $\beta z_{\theta,1}$, and losses of facilitation at level $i \in \{0, \ldots \theta\}$ occur at rate $\lambda z_{i,1}$. A slightly less formal but perhaps more intuitive way of describing the above dynamics follows, writing $z$ the network state at time $t$ in a matrix form (we show here the transpose in order to save space):
\[z^T \equiv \begin{bmatrix}
  z_{0,0} & z_{1,0} & \ldots & z_{i,0} & \ldots & z_{\theta,0} \\
  z_{0,1} & z_{1,1} & \ldots & z_{i,1} & \ldots & z_{\theta,1} 
\end{bmatrix}
\]
A \emph{loss of facilitation} can occur leading to (changes appear in red and the occurrence rate appears above the arrow):
\[z^T \, \overset{\textcolor{orange}{\lambda z_{i,1}}}{\rightarrow} \,
  \begin{bmatrix}
    z_{0,0} & z_{1,0} & \ldots & z_{i,0}\textcolor{red}{+ 1} & \ldots & z_{\theta,0} \\
    z_{0,1} & z_{1,1} & \ldots & z_{i,1}\textcolor{red}{- 1} & \ldots & z_{\theta,1} 
  \end{bmatrix}
\]  
Alternatively, an \emph{inefficient spike} will lead to:
\[z^T  \, \overset{\textcolor{blue}{\beta z_{\theta,0}}}{\rightarrow} \,
  \begin{bmatrix}
    z_{0,0} & z_{1,0} & \ldots & z_{i,0} & \ldots & z_{\theta,0} \textcolor{red}{- 1}\\
    z_{0,1} \textcolor{red}{+ 1} & z_{1,1} & \ldots & z_{i,1} & \ldots & z_{\theta,1} 
  \end{bmatrix}
\]
While an \emph{efficient spike} will give:
\[z^T  \, \overset{\textcolor{blue}{\beta z_{\theta,1}}}{\rightarrow} \,
  \begin{bmatrix}
    \textcolor{red}{0} & z_{0,0} & \ldots & z_{\theta,0} \textcolor{red}{+ z_{\theta-1,0}}\\
    \textcolor{red}{1} & z_{0,1}  &  \ldots & z_{\theta,1} \textcolor{red}{+ z_{\theta-1,1}-1}
  \end{bmatrix}
\]

In the sequel this will be expressed alternatively as an infinitesimal generator $Q$ in matrix form, or an infinitesimal generator $\mathcal{L}$ in functional form (in Section \ref{sec:approx}), depending on what is more suitable for the current purpose.

\section{Cutting the state space into pieces}

\label{sec:cutting}

In the next section we will be interested in studying the quasi-stationary distribution of our system, that is the stationary distribution associated with the Markov process obtained when conditioning $(Z_t)_{t \geq 0}$ on non-absorption. In order to ensure that this actually makes sense we have two concerns: \textit{(i)} we would like to properly define the absorbing region; \textit{(ii)} prove that the process restricted to the complement of this absorbing region is irreducible. Once this is established, the existence and uniqueness of the quasi-stationary distributions follow from classical results (see \cite{darroch}).

\subsection{Absorbing region}

First notice that there is an obvious absorbing set of states for $(Z_t)_{t \geq 0}$, which are the states $D \subset R$ defined by $D = \{z \in R: z_{\theta,\bullet} = 0 \text{ and } \sum_{i=0}^{\theta-1} z_{i,1} = 0\}$, that is the states with no neuron at level $\theta$ and no facilitated neurons. Then neither a spike, nor a facilitation loss can happen and the process stays there for eternity.

Nonetheless we would be short-sighted if we stopped there, and took this set $D$ to be the absorbing region. It is clear that one can find other states which, while not being properly absorbing, can only lead to $D$ in a bounded number of steps with probability one. We write $A_\theta = \{z \in R: z_{\theta,1} = 0\}$. Notice that $D \subset A_\theta$, moreover if $(Z_t)_{t \geq 0}$ reaches some $a'\in A_\theta$ at some point, then after some inefficient spikes and after the facilitated neurons loose their facilitation one after the other (which is the only thing that can happen) the process hits $D$, and again stays there for eternity. This happens in a finite number of transitions --- at most $2N$ of them actually, which corresponds to the case in which $a'_{\theta,0} = N$.  More generally, for any $i \in \{1, \ldots \theta\}$ define: $$A_i = \left\{z \in R: \sum_{j=i}^\theta z_{j,1} \leq \theta - i\right\}.$$
Notice that, for $i=\theta$, this agrees with the previous definition. Furthermore it is easy to see that if at some point the process $(Z_t)_{t \geq 0}$ reaches $A_i$ for some $i \in \{1, \ldots, \theta\}$, then with probability one it then reaches $A_\theta$ after a maximum of $\theta-i$ efficient spikes, and again reaches $D$ after some inefficient spikes and facilitation losses.

For the case $i=0$ we need to be a little bit more careful, since when non-facilitated neurons at level $\theta$ spike the number of facilitated neurons at level $0$ increases, which in turn could prevent the process from being absorbed. Therefore, we define $$A_0 = \left\{z \in R: z_{\theta,0} + \sum_{j=0}^\theta z_{j,1} \leq \theta\right\}.$$ Again it is easy to see that, if $(Z_t)_{t \geq 0}$ reaches $A_0$ at some point it then hits $A_\theta$ in a maximum of $\theta$ spikes and then gets absorbed. Finally we define the absorbing region to be: $$A = \bigcup_{i=0}^\theta A_i.$$

We claim that this is the correct definition of the absorbing region, meaning that:

\begin{enumerate}
\item once $A$ has been reached the system can't get out of it, and eventually gets definitively absorbed in $D$, which happens with probability one in a bounded number of steps\footnote{For example one can easily see from the considerations above that the number of transition to reach $D$ from any point in $A$ has to be less than $\theta +2N$, even if it is certainly not the sharpest possible bound.},
\item and if the system is in $A^c$ then it has a positive probability of staying out of $A$ for an arbitrarily big number of steps.
\end{enumerate}

The first point is simply a reformulation of the discussion above. The second point is a consequence of the important fact that, whenever the process starts outside $A$, there is always a positive probability of having an arbitrarily big number of (efficient) spikes in any time interval of arbitrary length,  with no facilitation loss in the meantime. We'll see in the next section that we can actually go further than that, and show that once the state space has been restricted a little bit, then when starting in $A^c$ any other state from $A^c$ can be reached.
\vspace{0.5 cm}

\subsection{Irreducibility}
\label{sec:irreducibility}
One can easily see that, even when $A$ has been taken out of the state space $R$, the process $(Z_t)_{t \geq 0}$ isn't properly irreducible. Consider for example the class $R' \subset R$ of states defined by $$R' = \left\{z \in R : z_{i,\bullet} = 0 \text{ for some } i \in \{0, \ldots, \theta-1\} \right\}.$$ 

Let $z \in R' \backslash A$ and suppose that $(Z_t)_{t \geq 0}$ starts from $z$. First suppose that $z_{0,\bullet} = 0$. Notice that if no spike occurs in the future, then all neurons loose their facilitation and soon the system reaches $A$. Otherwise, after the first spike the system reaches some state $z'$ which is such that $z'_{0,\bullet} \geq 1$, and then this inequality remains true for eternity, as the only way for the number of neurons at level $0$ to decrease is if there is an efficient spike, moving neurons from level $0$ to level $1$, but in this case the number of neurons at level $0$ is immediately reset to $1$ (because the membrane potential of the spiking neuron is reset to $0$). Now if we suppose that $z_{0,\bullet} \geq 1$ but $z_{1,\bullet} = 0$, then at the instant of the first efficient spike the system reaches a state $z'$ such that $z'_{1,\bullet} \geq 1$, and again, the only way the number of neurons at level $1$ can then be affected is when there is an efficient spike, which push the value at level $0$ to level $1$, so that it can never be less than $1$ (because we already know that at level $0$, there can't be $0$ neuron anymore). If both $z_{0,\bullet}$ and $z_{1,\bullet}$ are equal to $0$, then it takes at maximum two spikes (the second one being efficient) to obtain the same result. More generally it is easy to see how the same argument applies recursively to prove the following statement.

\begin{lemma}	
	let $z \in R' \backslash A$ and let $(Z^z_t)_{t \geq 0}$ be the process starting from $z$. Then either $(Z^z_t)_{t \geq 0}$ is absorbed in $A$ before reaching $R \backslash R'$, or it reaches $R \backslash R'$ and then never hits $R'$ again. In the second case $R \backslash R'$ is reached after a maximum of $\theta-1$ efficient spikes.
\end{lemma}

While a recent article \cite{champagnat} has treated the existence of a quasi-stationary measure for discrete-time Markov chains in such a case, that is when the state space is not properly irreducible but consists in two successive classes, to the best of our knowledge results are still lacking in the continuous-time framework. Nonetheless, as $R'$ is always left after a maximum of only $\theta-1$ efficient spikes, it might simply be discarded from the state space. Indeed what the lemma above is essentially saying is that states in $R'$ are atypical states, artifacts that are only possible if we force the system to start from there, which are soon left, with no return possibility. Thus we will simply get rid of $R'$ and consider the restricted state space $\widehat{R} = R \backslash R'$, and then define the subset $R^* = \widehat{R} \backslash A$, which is the support of $(Z_t)_{t \geq 0}$ previous to extinction. We are now set to prove the main result of this section.

\begin{proposition} \label{prop:irreducibility}
	$(Z_t)_{t \geq 0}$ is irreducible on $R^*$.
\end{proposition}

\begin{proof}
	Fix $x,y \in R^*$ and some $t>0$. We shall find a finite sequence of states $y^0, y^1, \ldots y^n \in R^*$ such that for $s>0$ we have $$\P \left(Z^x_s = y^0 \right) > 0,$$  
	$$\P \left(Z^{y^i}_s = y^{i+1} \right) > 0 \text{ for any } i \in \{0, \ldots k-1\},$$
	$$\text{and } \P \left(Z^{y^n}_s = y\right) > 0.$$
	
	Of course the exact value of $s$ is unimportant, as in our continuous time setting if this is true for some $s$ then it is true for any $s$. Once we've obtain this, the result evidently follows from Markov property: 
	
	\begin{align*}
	\P \left( Z^x_t = y\right) &\geq \P \left(Z^x_{\frac{t}{n+2}} = y^0, Z^{y^0}_{\frac{2t}{n+2}} = y^1,  \text{ } \ldots  \text{ }, Z^x_{\frac{(n+1)t}{n+2}} = y^n, Z^{y^n}_t = y\right) \\
	&= \P \left( Z^x_{\frac{t}{n+2}} = y^0\right) \P \left( Z^{y^0}_{\frac{t}{n+2}} = y^1\right) \ldots  \P \left(Z^{y^n}_{\frac{t}{n+2}} = y\right) > 0.
	\end{align*}

	\begin{center}
		\begin{figure}[htbp]
		\centering
		\includegraphics[width=\textwidth]{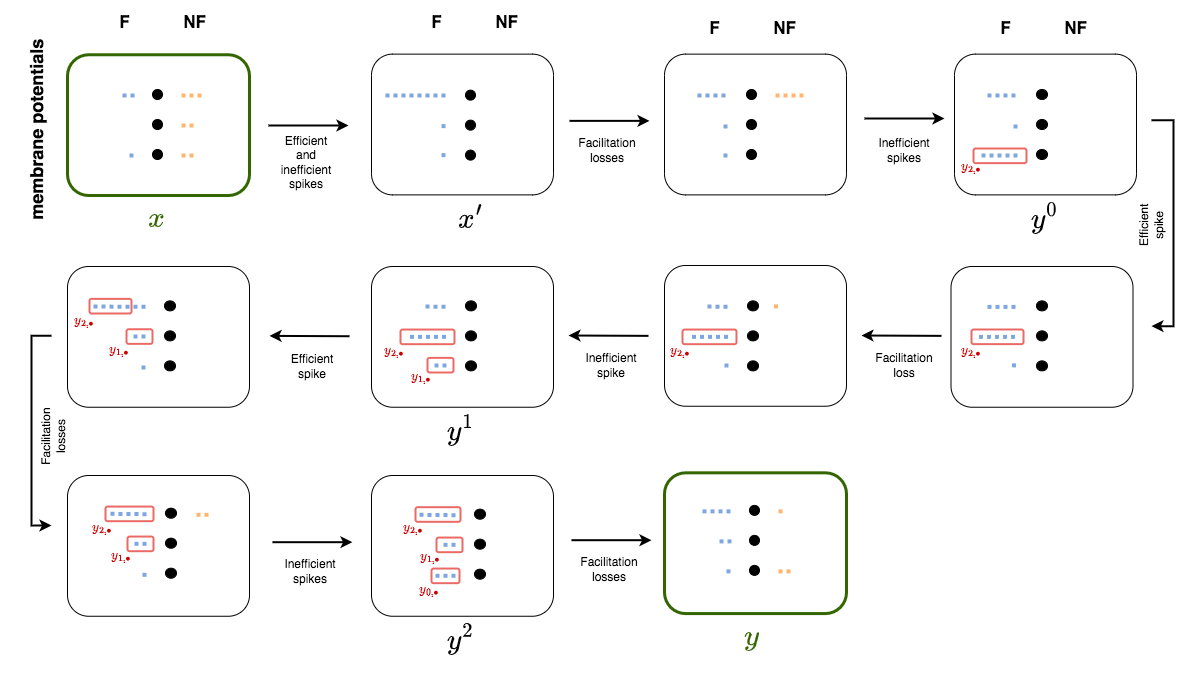}
		\caption{\label{fig:proof} Illustration of the idea behind the proof of Proposition \ref{prop:irreducibility} } in a minimal example with $N=10$ and $\theta=2$.
		\end{figure}
	\end{center}

	We now define our sequence. It is advisable to take a glance at Figure \ref{fig:proof} before reading the following text, which is a somewhat laborious (but necessary) translation of a nonetheless simple idea. For the sake of notation clarity the first state in the sequence will actually be written $x'$, the next state being $y_0$ and so on. $x'$ is the state in $R^*$ which is such that $x'_{\theta,1} = N - \theta$, and $x'_{i,1} = 1$ for all $i \in \{0, \ldots \theta-1\}$ (and of course all other coordinates to $0$). One can easily see that there is always a positive probability to reach $x'$ from $x$. Indeed suppose that, starting from $x$, the system undergoes exactly $\theta$ efficient spikes (and that nothing else happens in-between). This sequence of events has positive probability starting from $A^c$, moreover this has the effect of pushing all non-facilitated neurons to level $\theta$. Then suppose that all these non-facilitated neurons at level $\theta$ spike --- and thus become facilitated --- and that then the systems undergoes again exactly $\theta$ effective spikes (and that no facilitation loss happens in the mean time). This sequence of events is of positive probability and whatever $x$ is, the process ends up on $x'$.
	
	Now we define the other elements of the sequence. For any $k \in \{0, \ldots \theta-1\}$ the state $y^k$ is defined by: 
	$$y^k_{i,1} = y_{\theta-k+i,\bullet} \text{ for } i \in \{0, \ldots k\},$$
	$$y^k_{\theta,1} = N - \theta + k + 1 - \sum_{i=\theta-k}^\theta y_{i,\bullet},$$
	$$\text{and if } k < \theta-1 \text{ then } y^k_{i,1}=1 \text{ for } i \in \{k+1, \ldots \theta - 1\}.$$
	
	Moreover the penultimate state in our sequence (before $y$) is $y^\theta$, defined by $y^\theta_{i,1} = y_{i,\bullet}$ for $i \in \{0, \ldots \theta\}$. First we show that for any $k \in \{0, \ldots \theta\}$ the state $y^k$ is indeed outside $A$. In order to do this it is sufficient to show that the number of facilitated neurons at any level is greater or equal to $1$. The fact that $y \notin R'$ implies (by definition) that $y_{i,\bullet} \geq 1$ for any $i \in \{0, \ldots \theta\}$, and thus we have $y^\theta \notin A$ as well. Now pick some $k \in \{0, \ldots \theta-1\}$. For $i \in \{0, \ldots k\}$ we have $y^k_{i,1} \geq 1$ for the same reason, and for $i \in \{k+1, \ldots \theta-1\}$ we have $y^k_{i,1} = 1$. Finally to check $y^k_{\theta,1} \geq 1$ use again the fact that $y \notin R'$ and notice that therefore $$\sum_{i=\theta-k}^\theta y_{i,\bullet} = n - \sum_{i=0}^{\theta-k-1} y_{i,\bullet} \leq n - \theta + k.$$ This implies that $$y^k_{\theta,1} = n - \theta + k+1 - \sum_{i=\theta-k}^\theta y_{i,\bullet} \geq 1.$$ All levels are checked, and therefore $y^k \notin A$.  
	
	Let see how the system might reaches $y^0$ from $x'$. $y^0$ is the state such that $y^0_{\theta,1} = n - \theta - (y_{\theta,\bullet} - 1)$, $y^0_{1,1} = y^0_{2,1}= \ldots = y^0_{\theta-1,1} = 1$ and $y^0_{0,1} = y_{\theta,\bullet}$. Then it is clear that $y^0$ is reached from $x'$ if exactly $y_{\theta,\bullet} - 1$ neurons at level $\theta$ loose their facilitation and then emit a non-efficient spike one after the other. This happens with positive probability. The fact that there is enough facilitated neurons at level $\theta$ in $x'$ for this to happen follows from the fact that $y_{\theta,\bullet} = n - \sum_{i=0}^{\theta-1} y_{i,\bullet} \leq n - \theta$ (remember that $y \notin R'$, so that $y_{i,\bullet} \geq 1$ for any $i \in \{0, \ldots \theta-1\}$).
	
	Now let's se how we might go from $y^k$ to $y^{k+1}$ for any $k \in \{0, \ldots \theta-1\}$. We know that $y^k \notin A$ so that with positive probability there is an efficient spike occurring in the system before anything else. Let $z^k$ denote the state reached after this spike. The number of active neurons at level $\theta$ in the current state of the process cannot decrease, since one neuron was taken out (the spiking neuron) and at least one was added (exactly one for $k<\theta-1$ and possibly more if $k=\theta-1$), that is $z^k_{\theta,1} \geq y^k_{\theta,1}$. For the levels from $1$ to $\theta-1$ the values are simply pushed one-step upward, that is $z^k_{i,1} = y^k_{i-1,1}$ for any $i \in \{1, \ldots \theta-1\}$. Finally evidently $z^k_{0,1} = 1$. In other words we have $z^k_{i,1} = y^{k+1}_{i,1}$ for any $i \in \{1, \ldots \theta-1\}$ and to go from $z^k$ to $y^{k+1}$ it only remains to find a way to make it agree at level $0$ and $\theta$ as well. But this is easy using the same trick as in the previous paragraph: it suffices that the system undergoes exactly $y_{\theta-(k+1),\bullet}-1$ successive facilitation losses at level $\theta$ and then the same number of inefficient spikes. This obviously fixes the number of neurons at level $0$, then the number at level $\theta$ has no choice but to agree, as the different levels have to sum up to $n$. Moreover there is always enough facilitated neurons at level $\theta$ in $z^k$ for this to happen with positive probability, as shown by the following computation:
	
	\begin{align*}
		z^k_{\theta,\bullet} &= y^k_{\theta,\bullet} = n - \theta + k + 1 - \sum_{i=\theta-k}^\theta y_{i,\bullet}\\
		&= \sum_{i=0}^{\theta-k-1} y_{i,\bullet} + \sum_{i=\theta-k}^\theta y_{i,\bullet} - \theta + k + 1 - \sum_{i=\theta-k}^\theta y_{i,\bullet}\\
		&= \sum_{i=0}^{\theta-(k+1)} y_{i,\bullet} - \theta + k + 1\\
		&= y_{\theta - (k+1),\bullet} +  \sum_{i=0}^{\theta-(k+1)-1} y_{i,\bullet} - \theta + k + 1\\
		&\geq y_{\theta - (k+1),\bullet}.
	\end{align*}
	
	To obtain le last inequality we've used the fact that, as $y \notin R'$, the elements in the sum need to be greater or equal to $1$.
	
	It only remains to go from $y^\theta$ to $y$, but this is very easy: this will happen if there is exactly $y_{0,0}$ facilitation losse(s) at level $0$ in the process starting from $y^\theta$, then $y_{1,0}$ facilitation losse(s) at level $1$ and so on, with no spike in the mean time. This happens with positive probability and the proof is over. 

\end{proof}

\section{Quasi-stationary distribution}

\label{sec:QSD}

As a consequence of the irreducibility proven in the previous section, classical results\footnote{good references on the subject are \cite{darroch}, \cite{meleard:2012} and section 4.6 in \cite{kijima:97}} guarantee that there exists a unique probability measure $\mu$ supported by $R^*$, called the \emph{quasi-stationary distribution} of $(Z_t)_{t \geq 0}$, which is such that for any $z \in R^*$ and any $t \geq 0$: $$\P \left( Z^\mu_t = z \ | \ Z^\mu_t \notin A \right) = \mu (z).$$

Furthermore, not only the theory guarantees the existence and uniqueness of such probability measure, but it also provides a direct method to actually compute it. In order to explain how this is done we consider the infinitesimal generator $Q$ of $(Z_t)_{t \geq 0}$. Suppose we re-indexed the states by mapping the elements in $A$ into $\{1, \ldots, |A|\}$ and the ones of $R^*$ into $\{|A| + 1, \ldots, |A| + |R^*|\}$. Then the matrix $Q$ looks like
\begin{equation*}
Q =
\begin{pmatrix}
\textbf{A} & \textbf{O}	\\
\textbf{R} & \textbf{T}
\end{pmatrix},
\end{equation*}
where the matrix \textbf{A} is the sub-generator corresponding to the transitions occurring inside $A$, \textbf{T} is the matrix corresponding to the transitions occurring inside $R^*$, \textbf{R} is the matrix corresponding to the transitions occurring from $R^*$ to $A$, and $\textbf{O}$ is a matrix full of $0$ with $|A|$ rows and $|R^*|$ columns. As explained earlier we might as well identify $A$ to a single absorbing state and hence write
\begin{equation*}
	Q =
	\begin{pmatrix}
		0 & \textbf{0}	\\
		\textbf{r} & \textbf{T}
	\end{pmatrix},
\end{equation*}
where $\textbf{0}$ designates the null row vector of length $|R^*|$ and $\textbf{r}$ is a non-null column vector, of length $|R^*|$ as well. Then, following \cite{kijima:97} (pages 214 and 215), the appropriate version of Perron-Frobenius theorem ensures the existence of a left-eigenvector $\mu \in [0,1]^{|R^*|}$ associated to the Perron-Frobenius eigenvalue $\alpha_{PF}$:
\begin{equation} \label{eq:pfeigen}
\mu \textbf{T} = \alpha_{PF} \mu.
\end{equation}

We know moreover that $\alpha_{PF} \in \R$ and that it is the largest eigenvalue in the real part, that is $\text{Re}(\alpha_{PF}) > \text{Re}(\alpha)$ for any eigenvalue $\alpha \neq \alpha_{PF}$. On the other hand we also know that $\mu$ is unique up to a multiplicative constant, so that there exists only one vector satisfying both (\ref{eq:pfeigen}) and 
\begin{equation} \label{eq:normalizationeigenpf}
	\mu \textbf{1} = 1,
\end{equation} where $\textbf{1}$ denotes the column vector of length $|R^*|$ with only $1$'s. These facts clearly suggest a simple procedure to obtain a particular QSD:

\begin{enumerate}
	\item first compute the eigenvalues of the matrix \textbf{T} and keep only the one that is maximal in the real part,
	\item then obtain the left eigenvector associated to this Perron-Frobenius eigenvalue and normalize it so that it satisfies (\ref{eq:normalizationeigenpf}).
\end{enumerate}

However, as $n$ and $\theta$ grow, the computational cost of using such method for our model rapidly explodes, so that we ought to design a more efficient approximate method. This is actually the subject of Section \ref{sec:approx}. Now, writing $\textbf{T}(t)$ for the (lossy) transition matrix associated to the transitions occurring in $R^*$ up to time $t$ we have
$$ \textbf{T}(t) = \sum_{k=0}^\infty \frac{\textbf{T}^k t^k}{k!},$$ so that applying (\ref{eq:pfeigen}) and denoting $\gamma = - \alpha_{PF}$ gives
\begin{equation}\label{eq:PFtoGamma}
\mu \textbf{T}(t) = \sum_{k=0}^\infty \frac{\left(\alpha_{PF} t\right)^k}{k!} \cdot \mu = e^{-\gamma t} \mu.
\end{equation}

Then taking the sum of all the elements of both vectors on the left-hand side and on the right-hand side respectively in the equation above gives that for any $t \geq 0$ 
\begin{equation} \label{eq:exttime}
	\P \left(Z^\mu_t \notin A \right) = e^{-\gamma t}.
\end{equation}
In words, the time of extinction of the system when the initial state is chosen according to the QSD has an exponential distribution of some parameter $\gamma$. As $\gamma$ is simply $-\alpha_{PF}$, a by-product of the two-steps procedure defined earlier is that we immediately obtain the rate of the extinction time when the system start from the QSD. Moreover, a classical result states that not only $\mu$ is quasi-stationary, but it is also the unique Yaglom limit, meaning that for any states $x$ and $z$ in $R^*$ the following holds\footnote{See Theorem 4.27 in \cite{kijima:97}.}
\begin{equation} \label{eq:quasilimit}	
	\lim_{t \rightarrow \infty} \P \left( Z^x = z\ \ | \ Z^x_t \notin A\right) = \mu(z).
\end{equation}
In other words, whatever the distribution of the initial state, there is always \emph{relaxation} toward the QSD for the unabsorbed process.

This is good news, as one of our goal is to establish the metastable nature of our system, which requires that we show that the time of extinction, starting from any given state in $R^*$, is exponentially distributed. But thanks to (\ref{eq:exttime}) and (\ref{eq:quasilimit}) this will be established at the condition that we show that the relaxation time toward the QSD is negligible relative to the extinction time. This last point will be studied by numerical means.

\section{Solution for $N$ small}

\label{sec:nsmall}

A concrete example of an explicit QSD construction (Sec. \ref{sec:QSD}) for the model of Sec. \ref{sec:aggregated-process} is presented next. Since the size of model state space grows very rapidly with $N$ and $\theta$, we consider here rather small values for these parameters, namely $N=5$ and $\theta=1$.

\subsection{The Q matrix as a graph}
The state space has only 56 elements and all the states can be displayed as shown on Fig. \ref{fig:ntwN5T1}. On this figure the states are explicitly represented as $z^T$ (columns correspond to membrane potential values: 0 for the first column and 1 for the second; rows correspond to the synaptic facilitation value: 0 for the first row and 1 for the second). Each element of $z^T$ contains the ``headcounts'': how many of the five neurons of the network have the corresponding membrane potential and synaptic facilitation state. The thickness of the arrows represents the type of transition: synaptic facilitation loss for thin arrows and spike for thick ones. The states belonging to $R^*$ are shown with a light blue background on Fig. \ref{fig:ntwN5T1}. Each such state has at least one arrow leading to it from another ``blue'' state and at least one arrow leaving it for another ``blue'' state.  
\begin{center}
  \begin{figure}
    \includegraphics[width=1.0\textwidth]{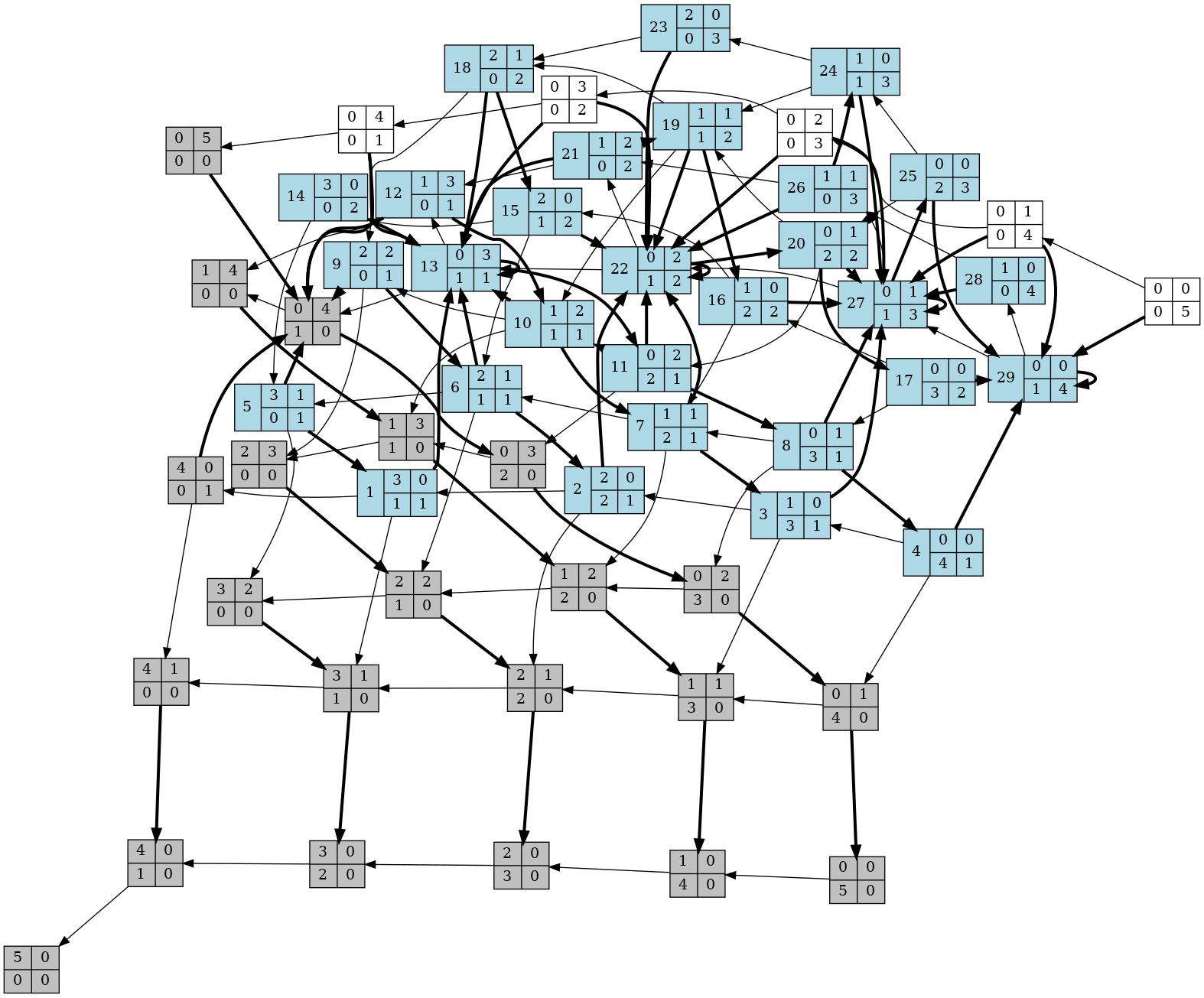}
    \caption{\label{fig:ntwN5T1}A network with 5 neurons and $\theta$ = 1. Configurations that belong to set $A$ appear in gray filled boxes, the ones that belong to set $R^*$ appear in light blue filled boxes with the index of the configuration on the left (in order to read the next matrix), while the ones belonging to set $R'\backslash A$ appear in white filled boxes. Transitions due to a spike appear as thick arrows, while transitions due to a synaptic facilitation loss appear as thin arrows.}
  \end{figure}
\end{center}

\pagebreak
\subsection{The \textbf{T} matrix}
Matrix \textbf{T} of Sec. \ref{sec:QSD} can be here explicitly represented as follows (the row and column indices correspond to the numbers of the ``blue'' states on Fig. \ref{fig:ntwN5T1}):
\begin{center}
  \tiny
\begin{tabular}{ccccccccccccccccccccccccccccc}
$\bullet$ & . & . & . & . & . & . & . & . & . & . & . & \textcolor{blue}{1} & . & . & . & . & . & . & . & . & . & . & . & . & . & . & . & . \\ 
\textcolor{orange}{2} & $\bullet$ & . & . & . & . & . & . & . & . & . & . & . & . & . & . & . & . & . & . & . & \textcolor{blue}{1} & . & . & . & . & . & . & . \\ 
. & \textcolor{orange}{3} & $\bullet$ & . & . & . & . & . & . & . & . & . & . & . & . & . & . & . & . & . & . & . & . & . & . & . & \textcolor{blue}{1} & . & . \\ 
. & . & \textcolor{orange}{4} & $\bullet$ & . & . & . & . & . & . & . & . & . & . & . & . & . & . & . & . & . & . & . & . & . & . & . & . & \textcolor{blue}{1} \\ 
\textcolor{blue}{1} & . & . & . & $\bullet$ & . & . & . & . & . & . & . & . & . & . & . & . & . & . & . & . & . & . & . & . & . & . & . & . \\ 
. & \textcolor{blue}{1} & . & . & \textcolor{orange}{1} & $\bullet$ & . & . & . & . & . & . & \textcolor{blue}{1} & . & . & . & . & . & . & . & . & . & . & . & . & . & . & . & . \\ 
. & . & \textcolor{blue}{1} & . & . & \textcolor{orange}{2} & $\bullet$ & . & . & . & . & . & . & . & . & . & . & . & . & . & . & \textcolor{blue}{1} & . & . & . & . & . & . & . \\ 
. & . & . & \textcolor{blue}{1} & . & . & \textcolor{orange}{3} & $\bullet$ & . & . & . & . & . & . & . & . & . & . & . & . & . & . & . & . & . & . & \textcolor{blue}{1} & . & . \\ 
. & . & . & . & . & \textcolor{blue}{2} & . & . & $\bullet$ & . & . & . & . & . & . & . & . & . & . & . & . & . & . & . & . & . & . & . & . \\ 
. & . & . & . & . & . & \textcolor{blue}{2} & . & \textcolor{orange}{1} & $\bullet$ & . & . & \textcolor{blue}{1} & . & . & . & . & . & . & . & . & . & . & . & . & . & . & . & . \\ 
. & . & . & . & . & . & . & \textcolor{blue}{2} & . & \textcolor{orange}{2} & $\bullet$ & . & . & . & . & . & . & . & . & . & . & \textcolor{blue}{1} & . & . & . & . & . & . & . \\ 
. & . & . & . & . & . & . & . & . & \textcolor{blue}{3} & . & $\bullet$ & . & . & . & . & . & . & . & . & . & . & . & . & . & . & . & . & . \\ 
. & . & . & . & . & . & . & . & . & . & \textcolor{blue}{3} & \textcolor{orange}{1} & $\bullet$ & . & . & . & . & . & . & . & . & . & . & . & . & . & . & . & . \\ 
. & . & . & . & \textcolor{orange}{2} & . & . & . & . & . & . & . & \textcolor{blue}{2} & $\bullet$ & . & . & . & . & . & . & . & . & . & . & . & . & . & . & . \\ 
. & . & . & . & . & \textcolor{orange}{2} & . & . & . & . & . & . & . & \textcolor{orange}{1} & $\bullet$ & . & . & . & . & . & . & \textcolor{blue}{2} & . & . & . & . & . & . & . \\ 
. & . & . & . & . & . & \textcolor{orange}{2} & . & . & . & . & . & . & . & \textcolor{orange}{2} & $\bullet$ & . & . & . & . & . & . & . & . & . & . & \textcolor{blue}{2} & . & . \\ 
. & . & . & . & . & . & . & \textcolor{orange}{2} & . & . & . & . & . & . & . & \textcolor{orange}{3} & $\bullet$ & . & . & . & . & . & . & . & . & . & . & . & \textcolor{blue}{2} \\ 
. & . & . & . & . & . & . & . & \textcolor{orange}{2} & . & . & . & \textcolor{blue}{2} & . & \textcolor{blue}{1} & . & . & $\bullet$ & . & . & . & . & . & . & . & . & . & . & . \\ 
. & . & . & . & . & . & . & . & . & \textcolor{orange}{2} & . & . & . & . & . & \textcolor{blue}{1} & . & \textcolor{orange}{1} & $\bullet$ & . & . & \textcolor{blue}{2} & . & . & . & . & . & . & . \\ 
. & . & . & . & . & . & . & . & . & . & \textcolor{orange}{2} & . & . & . & . & . & \textcolor{blue}{1} & . & \textcolor{orange}{2} & $\bullet$ & . & . & . & . & . & . & \textcolor{blue}{2} & . & . \\ 
. & . & . & . & . & . & . & . & . & . & . & \textcolor{orange}{2} & \textcolor{blue}{2} & . & . & . & . & . & \textcolor{blue}{2} & . & $\bullet$ & . & . & . & . & . & . & . & . \\ 
. & . & . & . & . & . & . & . & . & . & . & . & \textcolor{orange}{2} & . & . & . & . & . & . & \textcolor{blue}{2} & \textcolor{orange}{1} & $\bullet$ & . & . & . & . & . & . & . \\ 
. & . & . & . & . & . & . & . & . & . & . & . & . & . & . & . & . & \textcolor{orange}{3} & . & . & . & \textcolor{blue}{3} & $\bullet$ & . & . & . & . & . & . \\ 
. & . & . & . & . & . & . & . & . & . & . & . & . & . & . & . & . & . & \textcolor{orange}{3} & . & . & . & \textcolor{orange}{1} & $\bullet$ & . & . & \textcolor{blue}{3} & . & . \\ 
. & . & . & . & . & . & . & . & . & . & . & . & . & . & . & . & . & . & . & \textcolor{orange}{3} & . & . & . & \textcolor{orange}{2} & $\bullet$ & . & . & . & \textcolor{blue}{3} \\ 
. & . & . & . & . & . & . & . & . & . & . & . & . & . & . & . & . & . & . & . & \textcolor{orange}{3} & \textcolor{blue}{3} & . & \textcolor{blue}{1} & . & $\bullet$ & . & . & . \\ 
. & . & . & . & . & . & . & . & . & . & . & . & . & . & . & . & . & . & . & . & . & \textcolor{orange}{3} & . & . & \textcolor{blue}{1} & \textcolor{orange}{1} & $\bullet$ & . & . \\ 
. & . & . & . & . & . & . & . & . & . & . & . & . & . & . & . & . & . & . & . & . & . & . & . & . & \textcolor{orange}{4} & \textcolor{blue}{4} & $\bullet$ & . \\ 
. & . & . & . & . & . & . & . & . & . & . & . & . & . & . & . & . & . & . & . & . & . & . & . & . & . & \textcolor{orange}{4} & \textcolor{orange}{1} & $\bullet$ \\ 
\end{tabular}
\normalsize

\end{center}  
The zeros are not explicitly written but replace by a "$\bullet$". The factor with which $\lambda$ should be multiplied appear in \textcolor{orange}{orange}, while the factor with which $\beta$ should be multiplied appear in \textcolor{blue}{blue}. We see for instance that state 29 (last row) can make a transition to state 28 (last row and penultimate column) by a \emph{single} synaptic facilitation loss or to state 27 (last row and antepenultimate column) by one among 4 possible synaptic facilitation losses. These two transitions correspond to what is displayed on Fig. \ref{fig:ntwN5T1}. The actual values of the diagonal are given next:
\begin{center}
  \small
\begin{tabular}{cc | cc | cc}
index & rate & index & rate & index & rate\\ \hline 
1 & -\textcolor{orange}{2$\lambda$} -\textcolor{blue}{$\beta$} & 11 & -\textcolor{orange}{3$\lambda$} -\textcolor{blue}{3$\beta$} & 21 & -\textcolor{orange}{2$\lambda$} -\textcolor{blue}{4$\beta$}\\ 
2 & -\textcolor{orange}{3$\lambda$} -\textcolor{blue}{$\beta$} & 12 & -\textcolor{orange}{$\lambda$} -\textcolor{blue}{4$\beta$} & 22 & -\textcolor{orange}{3$\lambda$} -\textcolor{blue}{2$\beta$}\\ 
3 & -\textcolor{orange}{4$\lambda$} -\textcolor{blue}{$\beta$} & 13 & -\textcolor{orange}{2$\lambda$} -\textcolor{blue}{3$\beta$} & 23 & -\textcolor{orange}{3$\lambda$} -\textcolor{blue}{3$\beta$}\\ 
4 & -\textcolor{orange}{5$\lambda$} -\textcolor{blue}{$\beta$} & 14 & -\textcolor{orange}{2$\lambda$} -\textcolor{blue}{2$\beta$} & 24 & -\textcolor{orange}{4$\lambda$} -\textcolor{blue}{3$\beta$}\\ 
5 & -\textcolor{orange}{$\lambda$} -\textcolor{blue}{2$\beta$} & 15 & -\textcolor{orange}{3$\lambda$} -\textcolor{blue}{2$\beta$} & 25 & -\textcolor{orange}{5$\lambda$} -\textcolor{blue}{3$\beta$}\\ 
6 & -\textcolor{orange}{2$\lambda$} -\textcolor{blue}{2$\beta$} & 16 & -\textcolor{orange}{4$\lambda$} -\textcolor{blue}{2$\beta$} & 26 & -\textcolor{orange}{3$\lambda$} -\textcolor{blue}{4$\beta$}\\ 
7 & -\textcolor{orange}{3$\lambda$} -\textcolor{blue}{2$\beta$} & 17 & -\textcolor{orange}{5$\lambda$} -\textcolor{blue}{2$\beta$} & 27 & -\textcolor{orange}{4$\lambda$} -\textcolor{blue}{$\beta$}\\ 
8 & -\textcolor{orange}{4$\lambda$} -\textcolor{blue}{2$\beta$} & 18 & -\textcolor{orange}{2$\lambda$} -\textcolor{blue}{3$\beta$} & 28 & -\textcolor{orange}{4$\lambda$} -\textcolor{blue}{4$\beta$}\\ 
9 & -\textcolor{orange}{$\lambda$} -\textcolor{blue}{3$\beta$} & 19 & -\textcolor{orange}{3$\lambda$} -\textcolor{blue}{3$\beta$} & 29 & -\textcolor{orange}{5$\lambda$}\\ 
10 & -\textcolor{orange}{2$\lambda$} -\textcolor{blue}{3$\beta$} & 20 & -\textcolor{orange}{4$\lambda$} -\textcolor{blue}{3$\beta$} &  & \\ 
\end{tabular}
\normalsize

\end{center}
We see that matrix \textbf{T} is very sparse, having at most 3 positive non-diagonal elements per row.

\subsection{Case $N=5$, $\theta=1$, $\beta=10$, $\lambda=4$}
\label{sec:n5t1b10l4}
We compare now the theoretical/direct results provided by solving Eq. \eqref{eq:pfeigen} and \eqref{eq:normalizationeigenpf} with a straightforward simulation of the Markov process (see Sec.~\ref{sec:numerical-details} for implementation details). $10^5$ independent replicates were simulated, all starting with each of the 5 neurons in the susceptible state (U=1), with a facilitated synapse (F=1). The simulation of a given replicate was stopped when one of the $A$ states was reached or, if it survived till then, when time 4 was reached. Fig. \ref{fig:sec5_n5t1b10l4_fig_2} shows this comparison. The left part of the figure displays, on a log ordinate scale, the empirical survival function (a replicate is ``dead'' when it reaches one of the $A$ states), as well as the theoretical one provided by Eq. \eqref{eq:exttime}. The right part of the figure displays the empirical mean headcounts for each state that a single neuron can have---namely: $(U=0,F=0)$ (orange), $(U=0,F=1)$ (blue), $(U=1,F=0)$ (red) and $(U=1,F=1)$ (black)---, together with the expected values computed from the QSD (horizontal lines):
\[\mu_{i,j} = \int_{R^*} z_{i,j} d\mu = \sum_{z \in R^*}z_{i,j} \, \mu(z)\, , \quad i,j \in \{0,1\} \, ,\]
where $\mu(\bullet)$ is the QSD (Eq. \eqref{eq:pfeigen} and \eqref{eq:normalizationeigenpf}). Notice the \emph{fast relaxation} of the empirical means towards the corresponding expected values; remember that each replicate starts with the 5 neurons being susceptible ($U=1$) with a facilitated synapse, a state that is not in $R^*$ (see Fig.~\ref{fig:ntwN5T1}, rightmost state). Sec.~\ref{sec:additional-comparisons} presents additional simulations showing that the case we just illustrated is at least reasonably representative. 
\begin{center}
  \centering
  \begin{figure}[H]
    \includegraphics[width=0.8\textwidth]{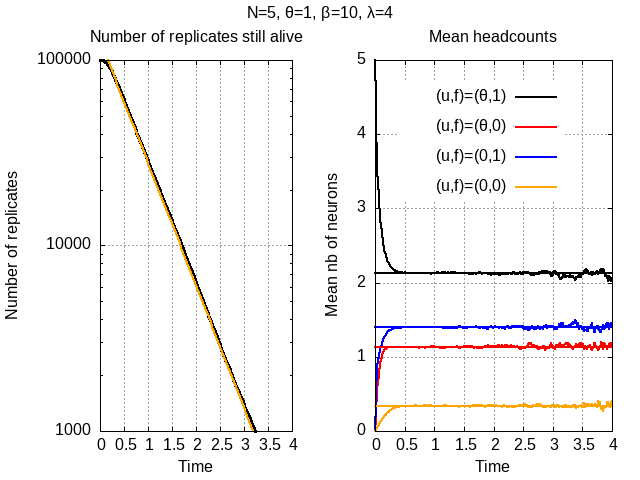}
    \caption{\label{fig:sec5_n5t1b10l4_fig_2}Stochastic simulations and direct solution comparison for a network with $N=5$, $\theta=1$, $\beta=10$, $\lambda=4$. $10^5$ replicates were used for the simulations. Left, the empirical survival function (black), together with the theoretical straight line (orange) whose slope is given by $\gamma$ in Eq. \eqref{eq:pfeigen}, \eqref{eq:PFtoGamma} and \eqref{eq:exttime}; right, the empirical mean headcounts of the four states (see legend), together with the expected values computed from the QSD (Eq. \eqref{eq:normalizationeigenpf}) shown as horizontal lines.}
  \end{figure}
\end{center}
Our stochastic simulations can be used for computing the empirical distributions, among the $R^*$ states, of the alive replicates at different times. This is what is shown on the upper part of Fig.~\ref{fig:emp-and-theo-qsd}. There a ``line plus glyph'' display is used instead of the vertical bars that would have been be more appropriate for a \emph{single} distribution. When looking at these graphs, the reader should therefore keep in mind that the only meaningful values are the ones indicated by the glyphs. The latter are located at integer values of the abscissa and correspond to the $R^*$ states indexed on Fig.~\ref{fig:ntwN5T1}. The lines in-between the glyphs are here to help seeing the \emph{profile} of the distributions: the distributions at the four different times are (essentially) a scaled copy of a \emph{fundamental} one, the QSD. The decaying values of the empirical distributions with time is just the consequence of replicates leaving the $R^*$ set for the $A$ set (Fig.~\ref{fig:sec5_n5t1b10l4_fig_2}, left). The profile identity of the empirical distributions is what is illustrated on the bottom part of Fig.~\ref{fig:emp-and-theo-qsd}: here the empirical distributions have been normalized to sum to 1, and the QSD as been added (thick red line).     
\begin{center}
  \centering
  \begin{figure}[H]
    \includegraphics[width=0.8\textwidth]{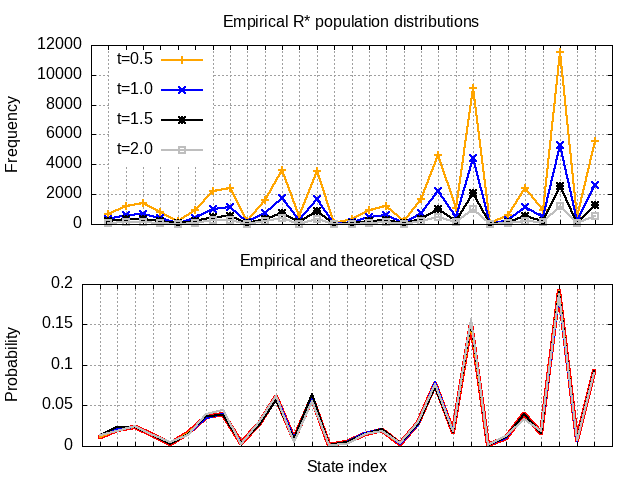}
    \caption{\label{fig:emp-and-theo-qsd}Stochastic simulations of a network with $N=5$, $\theta=1$, $\beta=10$, $\lambda=4$. $10^5$ replicates were used for the simulations. Top, empirical distributions of the still alive replicates among the 29 states of $R^*$ at four different times: 0.5, 1.0, 1.5 and 2.0. Bottom, scaled versions of the above distributions (they sum to 1) and theoretical QSD (thick red line).}
  \end{figure}
\end{center}
The key feature justifying the name of the \emph{quasi-stationary distribution} (QSD) is illustrated here: considering the \emph{alive replicates only}, the fraction of them in each state of $R^*$ \emph{does not change with time}. 

\section{Solution for $N$ big}

\label{sec:nbig}

When $N$ and $\theta$ grow, the number of states in $R^*$, $\binom{N + 2 \theta + 1}{2 \theta + 1}$, becomes very quickly, very large. Finding the Perron-Frobenius eigenvalue $\alpha_{PF}$ (Eq.~\ref{eq:pfeigen}) and the QSD (Eq.~\ref{eq:pfeigen} and \ref{eq:normalizationeigenpf}), as we did in the previous section, is not doable anymore and we have to stick to stochastic simulations. The direct comparison of the theoretical and simulation based results presented in the last section does nevertheless give us confidence in our stochastic simulation code --- we can add to that the simplicity of this code of course. We consider here networks with a fixed ratio $N/\theta=5$, $N=50$ and $N=500$. Like in the previous section, $10^5$ replicates are used and all replicates start with all the neurons in the susceptible state with a facilitated synapse. Fig.~\ref{fig:emp-survival-time} shows the increase of the survival time with the network size ($N$ at constant $N/\theta$ ratio). Looking at the graph we see that the level of $10^5 \, \exp -1$ (the level at one time constant) is reached at time $\sim 0.5$ with $N=5$, $\sim 1.5$ with $N=50$ and $\sim 3.8$ with $N=500$.     
\begin{center}
  \centering
  \begin{figure}[H]
    \includegraphics[width=0.8\textwidth]{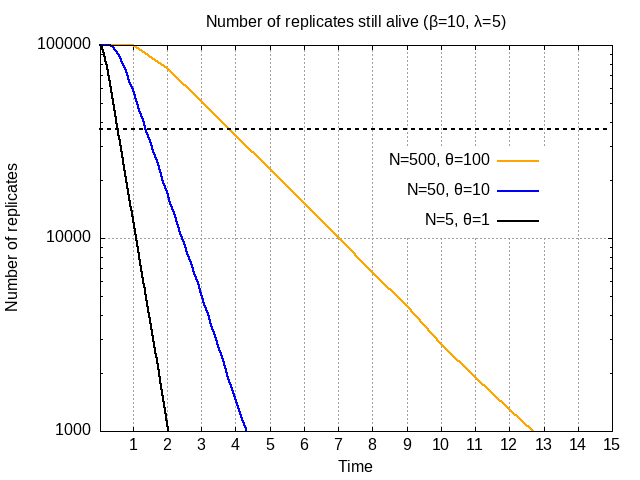}
    \caption{\label{fig:emp-survival-time}Empirical survival function from stochastic simulations of networks with $N=5, \theta=1$ (black), $N=50, \theta=10$ (blue), $N=500, \theta=100$ (orange) and $\beta=10$, $\lambda=5$. $10^5$ replicates were used for each parameter set. Each replicate with each parameter set started with all the neurons susceptible and all synapses facilitated. Dotted horizontal line, level at one time constant.}
  \end{figure}
\end{center}

On the other hand, as shown by Fig.~\ref{fig:relaxation}, the relaxation from the initial state, that is not in $R^*$, towards the QSD value is ``fast'' ($< 1$). It does not depend strongly on the network size (the empirical mean displayed on the figure for $N=50$ has been scaled to facilitate the comparison). This figure only shows the evolution of the mean number of susceptible neurons with a facilitated synapse in order to reduce clutter/improve readability\footnote{Showing all empirical means would require 22 curves for $\theta=10$ and 102 for $\theta=100$.}. Seeing this figure, the reader should abstain from jumping to the conclusion that the mean  QSD value, $\mu_{\theta,1}$, is proportional to the network size. As the approximate solution developed in the next section will make clear, this is only ``true'' when $\beta \, \mu_{\theta,1} \gg \lambda$ and this order of magnitude difference becomes more accurate as $N$ grows. Looking at the middle right panel of Fig.~\ref{fig:sec5_big_fig} in Sec.~\ref{sec:additional-comparisons} the reader can see that for $N=5$, the QSD mean value is closer to 2 (and would therefore be at 200 using the proper scaling on Fig.~\ref{fig:relaxation}). 
\begin{center}
  \centering
  \begin{figure}[H]
    \includegraphics[width=0.8\textwidth]{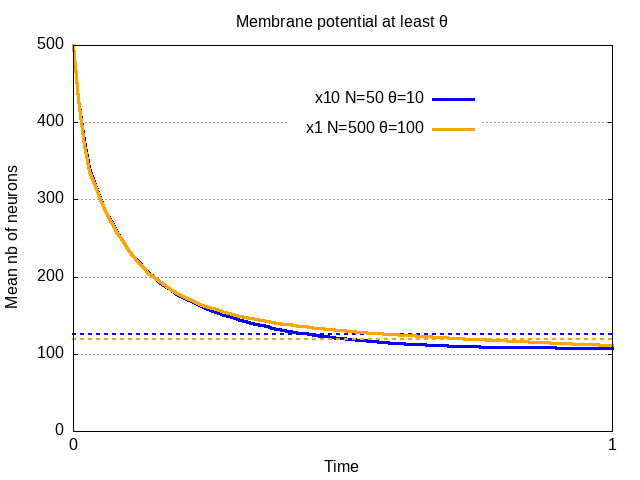}
    \caption{\label{fig:relaxation}Scaled mean number of susceptible neurons with a facilitated synapse as a function of time. Parameter used: $N=50, \theta=10$ (blue), $N=500, \theta=100$ (orange) and $\beta=10$, $\lambda=5$. $10^5$ replicates were used for each parameter set. Each replicate with each parameter set started with all the neurons susceptible and all synapses facilitated. The dotted lines show the approximate QSD mean values. They are the numerical solutions of the implicit Eq. \eqref{eq:conditionB} obtained in the next section.}
  \end{figure}
\end{center}

\section{Approximation of the first moment of the Quasi-stationary distribution}

\label{sec:approx}

\subsection{What is approximated?}
\label{sec:what-is-approximated}
The method presented in Section \ref{sec:QSD} to obtain the QSD of the system using Perron-Frobenius theorem becomes rapidly intractable as the number of neurons in the system grows. To detect the emergence of the QSD along a particular trajectory it might therefore be useful to design some less costly approximate method. In order to do so we would like to compute the expected values w.r.t. the QSD of the number of elements in each possible combinations of membrane potential and facilitation states. Remember that, for any $(i,j) \in \{0, \ldots \theta\} \times \{0,1\}$, $\mu_{i,j}$ denotes the QSD mean value: $$ \mu_{i,j} = \int_{R^*} z_{i,j} d \mu.$$   

We adopt the same notation as for elements of $R$ and for any $i \in \{0, \ldots \theta\}$ we write 
\begin{equation}\label{eq:sigma-def}
	\mu_{i,\bullet} = \mu_{i,0} + \mu_{i,1}.
\end{equation}

In the sequel we write the infinitesimal generator of $(Z_t)_{t \geq 0}$ as a functional operator $\mathcal{L}$, defined by:

\begin{equation} \label{eq:infinitesimalG}
	\mathcal{L} f(z) = \beta z_{\theta,1}  \big( f(\pi (z)) - f(z) \big) \ + \ \beta z_{\theta,0}  \big( f(\pi^* (z)) - f(z) \big) \ + \  \sum_{i=0}^\theta \lambda z_{i,1} \left( f(\pi^\dagger_i (z)) - f(z) \right)
\end{equation} where $f$ is any suitable test function and $z \in R$. In order to get the infinitesimal rate of change of the number of neuron at coordinate $(i,j) \in \{0, \ldots \theta\} \times \{0,1\}$ in $(Z_t)_{t \geq 0}$, starting from some state $z \in R$, we consider the functions from $R$ to $\{0, \ldots N\}$ defined by $f_{i,j} : z \mapsto f_{i,j}(z) = z_{i,j}$. We perform the computation for $f_{0,0}$. The second term in the right hand side of (\ref{eq:infinitesimalG}) equals $0$, as $f_{0,0}(\pi^*(z)) - f_{0,0}(z) = 0$ --- an inefficient spike doesn't changes the number of neurons at the coordinate $(0,0)$. In the third term, all the elements of the sum equals $0$ except for $i=0$, where $f_{0,0}(\pi^\dagger_0 (z)) - f_{0,0}(z) = 1$ --- as a facilitation loss affects the number of non-facilitated neurons at level $0$ only if it occurs on a neuron of null membrane potential, and when it does so it increases this number by one unit. For the first term, we have $f_{0,0}(\pi(z)) = 0$, and therefore $f_{0,0}(\pi(z)) - f_{0,0}(z) = -z_{0,0}$. Finally the complete expression is 
\begin{equation} \label{eq:infG00}
	\mathcal{L}f_{0,0}(z) = - \beta z_{\theta,1} z_{0,0} + \lambda z_{0,1}.
\end{equation}

Now, if $\mu$ were a stationary distribution for $(Z_t)_{t \geq 0}$, by standard results on Markov processes\footnote{See for example theorem 3.37 in \cite{liggettMarkovProcesses}.} we should have $\int \mathcal{L}f_{i,j} d \mu = 0$. Of course, strictly speaking, $\mu$ is not a stationary distribution for $(Z_t)_{t \geq 0}$, as it is only a quasi-stationary distribution, or in other words a stationary distribution for the modification of $(Z_t)_{t \geq 0}$ in which we've prohibited any transition to the absorbing region $A$. We will assume nonetheless here that $\mu$ is close enough to be a stationary distribution for $(Z_t)_{t \geq 0}$ and try to solve $\int \mathcal{L}f_{i,j} d \mu = 0$ for $\mu_{i,j}$. Using (\ref{eq:infG00}), this gives $$0 = - \beta \int z_{\theta,1} z_{0,0} d \mu + \lambda \mu_{0,1}.$$

We proceed to further approximations by assuming that $\int z_{\theta,1} z_{0,0} d \mu \approx \mu_{\theta,1} \mu_{0,0}$, and we obtain the following equation

\begin{equation}\label{eq:E_mu_0}
	0 = -\beta \mu_{\theta,1} \mu_{0,0} + \lambda \mu_{0,1}
\end{equation}

\vspace{0.4 cm}

\subsection{Getting the approximate expected values of the QSD}
\label{sec:getting-the-approximate-expected-values-of-the-qsd}
Proceeding by similar arguments for the others $f_{i,j}$ we obtain the following equations.

\begin{align}
	0 &= -\beta \, \mu_{\theta , 1} \, \mu_{0,1} - \lambda \, \mu_{0,1} + \beta \, \mu_{\theta,\bullet}\, . \label{eq:E_mu_1} 
\end{align}
For $i \in \{1,\ldots,\theta-1\}$:
\begin{align}
	0 &= -\beta \, \mu_{\theta , 1} \, (\mu_{i,0} - \mu_{i - 1,0}) + \lambda \, \mu_{i ,1} \, , \label{eq:E_mu_2i}\\
	0 &= -\beta \, \mu_{\theta, 1} \, (\mu_{i, 1} - \mu_{i - 1,1}) - \lambda \, \mu_{i , 1} \, . \label{eq:E_mu_2ip1}
\end{align}
And for $\mu_{\theta,0}$ and $\mu_{\theta, 1}$:
\begin{align}
	0 &= -\beta \, \mu_{\theta,0} + \beta \, \mu_{\theta, 1} \, \mu_{\theta - 1,0} + \lambda \, \mu_{\theta , 1} \, , \label{eq:E_mu_2theta}\\
	0 &= -\beta \, \mu_{\theta, 1} + \beta \, \mu_{\theta, 1} \, \mu_{\theta - 1,1} - \lambda \, \mu_{\theta, 1} \, . \label{eq:E_mu_2thetap1}
\end{align}

We introduce:
\begin{equation}
	\rho_i = \frac{\mu_{ i , 1}}{\mu_{i,\bullet}}\, , \quad i = 0,\ldots,\theta\, . \label{eq:rho-def} 
\end{equation}
$\mu_{i,\bullet}$ is the total number of neurons whose membrane potential is $i$ in the network, regardless of their synaptic facilitation status. $\rho_i$ is the faction of neurons with a facilitated synapse among the $\mu_{i,\bullet}$ neurons whose membrane potential is $i$. We have obviously: $\mu_{i , 1} = \mu_{i,\bullet} \, \rho_i$ and $\mu_{i,0} = \mu_{i,\bullet} \, (1-\rho_i)$.
If we add Eq. \ref{eq:E_mu_0} and \ref{eq:E_mu_1} we get with Eq. \ref{eq:sigma-def} and \ref{eq:rho-def}:
\[0 = - \beta \mu_{\theta,\bullet} \, (\rho_{\theta} \mu_{0,\bullet} - 1) \,.\]
The QSD must have a non null value of $\mu_{\theta,1}$ and therefore of $\mu_{\theta,\bullet}$ otherwise it would be identical to the absorbing state. In a more useful way we can then write the last equation as:
\begin{equation}
	\rho_{\theta} \mu_{0,\bullet} = 1 \, . \label{eq:sigma0-condition} 
\end{equation}
If we add Eq. \ref{eq:E_mu_2i} and \ref{eq:E_mu_2ip1} we get with Eq. \ref{eq:sigma-def}:
\begin{equation}
	\mu_{i,\bullet} = \mu_{i-1,\bullet} \quad i=1,\ldots,\theta-1 \, . \label{eq:sigmaI-condition} 
\end{equation}
If we add Eq. \ref{eq:E_mu_2theta} and \ref{eq:E_mu_2thetap1} we get with Eq. \ref{eq:sigma-def}:
\begin{equation}
	\rho_{\theta} \mu_{\theta-1,\bullet} = 1 \, . \label{eq:sigmaTheta-condition} 
\end{equation}
Eq. \ref{eq:sigmaI-condition} implies that we can define $\kappa$ such that:
\begin{equation}
	\kappa = \mu_{0,\bullet} = \mu_{1,\bullet} = \cdots = \mu_{\theta-1,\bullet} \, . \label{eq:kappa-def}
\end{equation}
In words, all the states below $\theta$ are equally populated. Since $N = \sum_{i=0}^{\theta} \mu_{i,\bullet}$, the QSD must satisfy:
\begin{equation}
	N = \kappa \theta + \mu_{\theta,\bullet} \, . \label{eq:N-kappa-theta-condition}
\end{equation}
Combining Eq. \ref{eq:sigma0-condition} or \ref{eq:sigmaTheta-condition} with Eq. \ref{eq:kappa-def}, we get:
\begin{equation}
	\rho_{\theta} = \mu_{\theta,1}/\mu_{\theta,\bullet} = 1 / \kappa \, . \label{eq:rho_theta-condition}
\end{equation}
Using Eq. \ref{eq:sigma-def}, \ref{eq:rho-def}, \ref{eq:kappa-def} and \ref{eq:rho_theta-condition}, we can rewrite Eq. \ref{eq:E_mu_1} as:
\[0 = -\beta \mu_{\theta,1} \rho_0 \kappa - \lambda \rho_0 \kappa + \beta \mu_{\theta,\bullet}\, ,\]
or
\[\rho_0 = \frac{\beta \mu_{\theta,\bullet}}{(\lambda + \beta \mu_{\theta,1}) \kappa}\, ,\]
but $1/\kappa = \mu_{\theta,1}/\mu_{\theta,\bullet}$ (Eq. \ref{eq:rho_theta-condition}), therefore:
\begin{equation}
	\rho_0 = \frac{\beta \mu_{\theta,1}}{ \lambda + \beta \mu_{\theta,1}} \, . \label{eq:rho_0}
\end{equation}
Using Eq. \ref{eq:sigma-def}, \ref{eq:rho-def}, \ref{eq:kappa-def} and \ref{eq:rho_theta-condition}, we can rewrite Eq. \ref{eq:E_mu_2ip1} as:
\[0 = -\beta \mu_{\theta,1} (\rho_i -  \rho_{i-1}) \kappa - \lambda \kappa \rho_i \, , \quad i = 1,\ldots,\theta-1\, ,\]
leading to:
\begin{equation}
	\rho_i = \frac{\beta \mu_{\theta,1}}{ \lambda + \beta \mu_{\theta,1}} \rho_{i-1}\, ,\quad i = 1,\ldots,\theta-1 \, . \label{eq:rho_i}
\end{equation}
Combining Eq. \ref{eq:rho_0} and Eq. \ref{eq:rho_i}, we get:
\begin{equation}
	\rho_i = \left(\frac{\beta \mu_{\theta,1}}{ \lambda + \beta \mu_{\theta,1}}\right)^{i+1}\, ,\quad i = 0,\ldots,\theta-1 \, . \label{eq:rho_iB}
\end{equation}
Now using Eq. \ref{eq:sigma-def}, \ref{eq:rho-def}, \ref{eq:kappa-def} and \ref{eq:rho_theta-condition}, we can rewrite Eq. \ref{eq:E_mu_2thetap1} as:
\begin{equation}
	\rho_{\theta-1} = \frac{\lambda + \beta}{\kappa \beta} \, . \label{eq:rho_theta_1}
\end{equation}
Eq. \ref{eq:rho_iB} and \ref{eq:rho_theta_1} together lead to:
\begin{equation}
	\left(\frac{\beta \mu_{\theta,1}}{\lambda + \beta \mu_{\theta,1}}\right)^{\theta} = \frac{\lambda + \beta}{\kappa \beta}\, . \label{eq:conditionA}
\end{equation}
But Eq. \ref{eq:N-kappa-theta-condition} and \ref{eq:rho_theta-condition} imply that:
\begin{equation}
	\frac{1}{\kappa} = \frac{\theta + \mu_{\theta,1}}{N} \, , \label{eq:one-over-kappa}
\end{equation}
leading to:
\begin{equation}
	\mu_{\theta,1} = \frac{N \beta}{\lambda + \beta} \, \left(\frac{ \beta \mu_{\theta,1}}{\lambda + \beta \mu_{\theta,1} }\right)^{\theta} - \theta\, . \label{eq:conditionB}
\end{equation}
$\mu_{\theta,1}$ must be solution of Eq. \ref{eq:conditionB}. If we find such a $\mu_{\theta,1}$, we immediately get $\kappa$ (Eq. \ref{eq:one-over-kappa}) and $\mu_{\theta,\bullet}$, the number of neurons at or above the spiking threshold $\theta$  (Eq. \ref{eq:rho_theta-condition}). With Eq. \ref{eq:rho_iB}, we get the successive $\mu_{i,0}$ and $\mu_{i,1}$ for $i = 0, \ldots, \theta-1$, that is, the whole description of the QSD in terms of its expected values. A comparison between this approximation for $\mu_{\theta,1}$ and empirical values computed from stochastic simulations is shown on Fig.~\ref{fig:relaxation} in the previous section. 

\subsection{Remarks}

The $\mu_{\theta,1}$ solution(s) of Eq.~\ref{eq:conditionB} depend on $\beta$ and $\lambda$ only through their ratio: $\eta = \lambda / \beta$. So multiplying $\beta$ and $\lambda$ by the same factor does not change the solution(s) of Eq.~\ref{eq:conditionB}.

We can consider limiting cases to check if the argument we just developed makes sense.

A network without facilitation loss ($\lambda{}=0$) is an instance of a limiting case. For such a network, as soon as a neuron has spiked, its synapse gets facilitated and remains facilitated forever. Then Eq.~\ref{eq:conditionB} leads to, $\mu_{\theta,1} = N -\theta{}$; Eq.~\ref{eq:one-over-kappa} gives, $\kappa = 1$ and Eq.~\ref{eq:rho_0} implies that, $\mu_{\theta,0} = 0$. The use of Eq.~\ref{eq:kappa-def} together with Eq.~\ref{eq:rho_iB}, yields: $\mu_{i,0} = 0$ and $\mu_{i,1}=1$ for $i = 0,\ldots,\theta{}-1$. We have therefore one neuron (with a facilitated synapse) at each possible membrane potential value below threshold and all the other neurons (with a facilitated synapse) at threshold. This is the state $x'$ defined in Sec.~\ref{sec:irreducibility} and in Fig.~\ref{fig:proof} and this is precisely what we expect from the argument developed in that section.

If there is no threshold ($\theta=0$), every neuron of the network has two accessible states: facilitated synapse and un-facilitated synapse; it goes from the first to the second with rate $\beta$ and from the second to the first with rate $\lambda$. Elementary Markov process results tell us that the probability of finding a neuron in the facilitated state is $\frac{\beta}{\beta{}+\lambda{}}$ leading to $\mu_{0,1} = \frac{N \beta{}}{\beta{}+\lambda{}}$ as given by Eq.~\ref{eq:conditionB} when $\theta{} = 0$.

\subsection{Back-of-the-envelope calculation leading to $\rho_{\theta}$}

We can get an approximate value of $\rho_{\theta} = \frac{\theta + \mu_{\theta,1}}{N}$ (Eq.~\ref{eq:sigma0-condition} and \ref{eq:one-over-kappa}) with an even bolder approach. Quantity $\beta \mu_{\theta,1}$ is the rate of efficient spikes; spike making all the other neurons climb one step on the membrane potential ladder. A neuron that just spiked has a null membrane potential and a facilitated synapse. That neuron needs $\theta$ efficient spikes in order to be able to spike again and that will require on average $\theta/(\beta \mu_{\theta,1})$ time units. Once the neuron membrane potential reaches $\theta$, that neuron can spike again and will do so with a rate $\beta$. The question is: what is the ``probability'' for the synapse of that neuron to be still facilitated when its next spike comes? In other words, what is the value of $\rho_{\theta}$? Facilitation is lost with rate $\lambda$, that probability is therefore:
\[\int_0^{\infty} e^{-\lambda \left(\theta/(\beta \mu_{\theta,1})+t\right)} \beta e^{-\beta t} dt = \frac{\beta}{\lambda + \beta} \exp \left(-\frac{\lambda \theta}{\beta \mu_{\theta,1}}\right)\, .\]
Now assuming that $\beta \mu_{\theta,1} \gg \lambda$, writing
\[\exp \left(-\frac{\lambda \theta}{\beta \mu_{\theta,1}}\right) = \left(\frac{1}{exp\frac{\lambda }{\beta \mu_{\theta,1}}}\right)^{\theta}\]
and using a first order Taylor expansion for $\exp(\bullet)$ leads to:
\[\rho_{\theta} = \frac{\beta}{\lambda + \beta} \, \left(\frac{ \beta \mu_{\theta,1}}{\lambda + \beta \mu_{\theta,1} }\right)^{\theta}\]
that is what we already obtained (we just have to combine Eq.~\ref{eq:conditionB}, \ref{eq:one-over-kappa} and \ref{eq:sigma0-condition}).

\subsection{How good is this approximation}

We can first compare the first moments of the QSD derived in Sec.~\ref{sec:getting-the-approximate-expected-values-of-the-qsd} with the exact values provided by the complete QSD of Sec.~\ref{sec:QSD} for the case illustrated in Sec.~\ref{sec:n5t1b10l4}: $N=5$, $\theta=1$, $\beta=10$, $\lambda=4$. The exact results are (also shown as horizontal lines on the right panel of Fig.~\ref{fig:emp-and-theo-qsd}):
\[\mu_{0,0}=0.342, \; \mu_{0,1}=1.398, \; \mu_{1,0}=1.135, \; \mu_{1,1}=2.125\]
The approximate values are:
\[\mu_{0,0}=0.285, \; \mu_{0,1}=1.400, \; \mu_{1,0}=1.347, \; \mu_{1,1}=1.968\]

Next, we only look at $\mu_{\theta,1}$ for networks with $\beta=10$, $\lambda=5$, a ratio $N/\theta=5$ (first table) or $N/\theta=10$ (second table) and a sequence of $(N,\theta)$ pairs. We compare the output of stochastic simulation with the approximate calculation. Our numerical solution of the implicit Eq.~\ref{eq:conditionB} gives us 4 significant digits after the decimal point (Sec.~\ref{sec:solving-the-implicit-eq}). The simulation results are shown with their associated standard error. $10^5$ replicates were used for the first 2 entries of each table, $10^4$ for the third and $5\times 10^3$ for the fourth. Three time units were simulated and the mean values computed at time 2 are reported in the tables. We used fewer replicates with the larger $N$ values in order to save time. With $N/\theta=5$ we observe:
\begin{center}
  \begin{tabular}{c | c | c}
    $(N,\theta)$ & Approx. & Simulation \\
    \hline
    (50,10) & 12.563 & 10.76 $\pm$ 0.05 \\
    (100,20) & 24.526 & 20.20 $\pm$ 0.06 \\
    (500,100) & 119.738 & 101.4 $\pm$ 0.5 \\
    (1000,200) & 238.661 & 212.3 $\pm$ 0.9 \\
    \hline
  \end{tabular}
\end{center}
We see that in this setting the relative approximation error is about 20\%. If we change slightly the setting, $N/\theta = 10$, in order to get higher spiking rates, we observe:
\begin{center}
  \begin{tabular}{c | c | c}
    $(N,\theta)$ & Approx. & Simulation \\
    \hline
    (50,5) & 25.216 & 24.91 $\pm$ 0.02 \\
    (100,10) & 50.400 & 50.14 $\pm$ 0.02 \\
    (500,50) & 251.866 & 251.8 $\pm$ 0.2 \\
    (1000,100) & 503.700 & 503.6 $\pm$ 0.3 \\
    \hline
  \end{tabular}
\end{center}
The standard errors are smaller because: i) fewer replicates reach the ``dead set'' $A$ per time units---the statistics are therefore computed with more observations---; ii) when the spiking rate is larger, the QSD is less spread around its first moments (its standard deviation is smaller) and our key approximation (Sec.~\ref{sec:what-is-approximated}) is more accurate. We see that when the network spiking rate is large enough (because $N/\theta$ is large as shown here, but that also holds if $\lambda$ is ``small''), the approximate solution becomes very close to the truth. 

\section{Conclusion}

We proposed a simple stochastic model of strongly interconnected neurons exhibiting synaptic facilitation, and suggested that the metastable properties of such system might give an interesting explanation of the mechanism behind the sustained activity observed experimentally in networks of neurons involved in short-term memory.

The metastability of this system was established using arguments based on the the notion of quasi-stationarity. These arguments are partially based on simulations and are therefore non-rigorous from the mathematical point of view. Nonetheless it allows a very simple and intuitive comprehension of the phenomenon of metastability: the system started in any given state rapidly converges to the Yaglom limit, which happens to be the QSD as well, and, once there, it evolves in an almost stationary manner for an exponentially distributed random time, due to classical results about QSD. The fact that the relaxation toward the QSD is much faster than the extinction of the system implies that the extinction time is exponentially distributed as well --- asymptotically with respect to the number of neurons --- starting from any state out of the absorbing region, i.e. genuine metastability. Establishing that fact is actually the only missing piece to make the argument rigorous, and it shall be noticed that such arguments would then apply to a large class of Markov processes (at least any irreducible Markovian system having an absorbing state) --- including the famous Contact process for example, which has been rigorously proven to exhibit metastability by various means \cite{schonmann,cassandro,schonmann:1988,mountford:1993}, but with technical complications making the understanding of the fundamental reasons behind the phenomenon difficult to interpret.

\pagebreak
\appendix
\section{Some numerical details}
\label{sec:numerical-details}
This section provides a brief outline of the numerical methods implemented to simulate our model, to ``solve the QSD'' (Eq.~\ref{eq:pfeigen} and \ref{eq:normalizationeigenpf}) and the implicit equation \eqref{eq:conditionB}. The actual codes and their comprehensive documentation can be found on our \texttt{GitLab} repository\footnote{\href{https://gitlab.com/c_pouzat/metastability-in-a-system-of-spiking-neurons-with-synaptic-plasticity}{https://gitlab.com/c\_pouzat/metastability-in-a-system-of-spiking-neurons-with-synaptic-plasticity}}.

\subsection{Simulations}
\label{sec:simulations}
The simulation of the Markov processes defined in Sec.~\ref{sec:formal-definition-as-a-continuous-time-markov-chain} and \ref{sec:aggregated-process} is straightforward once the intensity matrix Q is known. Our \texttt{Fortran} code implements algorithm 2.7.2 of \cite[p. 80]{rubinstein.kroese:16}. This algorithm is just the constructive definition of a Markov process of \cite[Sec. 8.3.2, p. 243]{iosifescu:07}. Our simulation code requires a uniform (pseudo)random number generator; we used the implementation of the \texttt{xoshiro128plus} generator \cite{blackman.vigna:21} provided by Jannis Teunissen's \texttt{rng\_fortran}\footnote{\href{https://github.com/jannisteunissen/rng_fortran}{https://github.com/jannisteunissen/rng\_fortran}.}.
\subsection{QSD computation}
\label{sec:QSD-computation}
The eigenvalues and left eigenvectors of $\mathbf{T}$ (Eq.~\ref{eq:pfeigen} and \ref{eq:normalizationeigenpf}) were obtained with subroutine \texttt{dgeev} of LAPACK-3.11.0\footnote{\href{https://www.netlib.org/lapack/}{https://www.netlib.org/lapack/}.}. Since the LAPACK code is written in FORTRAN 77, we used the modern Fortran interface provided on \texttt{GitHub} by the \emph{Numerical Algorithms Group} (NAG)\footnote{\href{https://github.com/numericalalgorithmsgroup/LAPACK_Examples}{https://github.com/numericalalgorithmsgroup/LAPACK\_Examples}.}.

\subsection{Solving the implicit Eq.~\eqref{eq:conditionB}}
\label{sec:solving-the-implicit-eq}
We used a modification of Arjen Markus robust Newton method\footnote{\texttt{robust\_newton.f90}: \href{https://flibs.sourceforge.net/robust_newton.f90}{https://flibs.sourceforge.net/robust\_newton.f90}.} \cite{markus:12}. This code explicitly brackets the roots of the equation:
\[\mu_{\theta,1} + \theta - \frac{N \beta}{\lambda + \beta} \, \left(\frac{ \beta \mu_{\theta,1}}{\lambda + \beta \mu_{\theta,1} }\right)^{\theta} \]
and provides therefore an upper bound of its error.

\section{Additional comparisons between theoretical QSD and simulations}
\label{sec:additional-comparisons}
We consider again the setting of Sec.~\ref{sec:nsmall}, with $N=5$, $\theta=1$, $\beta=10$. As a complement to the case $\lambda=4$ of Sec.~\ref{sec:n5t1b10l4}, we illustrate here cases where $\lambda$ take values 3, 5 and 8. The equivalent of Fig.~\ref{fig:sec5_n5t1b10l4_fig_2} is then obtained and shown on Fig.~\ref{fig:sec5_big_fig}. The basic features are the same as what is seen on Fig.~\ref{fig:sec5_n5t1b10l4_fig_2}. Not surprisingly, the mean survival depends strongly on $\lambda$; the smaller the latter, the larger the former. The mean number of neurons in the susceptible state with a facilitated synapse (black curves on the right panels) exhibits a similar trend. The ``noisy'' aspect of the empirical means displayed on  Fig.~\ref{fig:sec5_big_fig} bottom, right ($\lambda=8$) for a time larger than 1.25 are due to the small sample size: few replicates are still alive past that time.
\begin{center}
  \centering
  \begin{figure}[H]
    \includegraphics[width=0.8\textwidth]{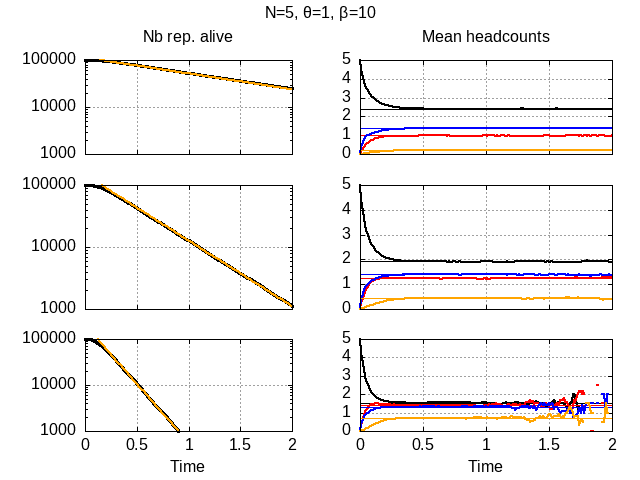}
    \caption{\label{fig:sec5_big_fig}Stochastic simulations and direct solution comparison for networks with $N=5$, $\theta=1$, $\beta=10$ and $\lambda=3$ (top), $\lambda=5$ (middle), $\lambda=8$ (bottom). $10^5$ replicates were used for the simulations. Left, the empirical survival function (black), together with the theoretical straight line (orange) whose slope is given by $\gamma$ in Eq. \eqref{eq:pfeigen}, \eqref{eq:PFtoGamma} and \eqref{eq:exttime}; right, the empirical mean headcounts of the four states (see legend), together with the expected values computed from the QSD (Eq. \eqref{eq:normalizationeigenpf}) shown as horizontal lines.}
  \end{figure}
\end{center}

\pagebreak
\bibliographystyle{plain}
\bibliography{Andre_Pouzat_MSSNSP_3}

\begin{thebibliography}{10}

\bibitem{amit.brunel:97}
D.~Amit and N.~Brunel.
\newblock Model of global spontaneous activity and local structured activity
  during delay periods in the cerebral cortex.
\newblock {\em Cerebral Cortex}, 7(3):237–252, Apr 1997.

\bibitem{barak.tsodyks:07}
O.~Barak and M.~Tsodyks.
\newblock Persistent activity in neural networks with dynamic synapses.
\newblock {\em PLOS Computational Biology}, 3(2):1--1, 02 2007.

\bibitem{barak.tsodyks:14}
O.~Barak and M.~Tsodyks.
\newblock Working models of working memory.
\newblock {\em Current Opinion in Neurobiology}, 25:20–24, Apr 2014.

\bibitem{bianchi:2016}
A.~Bianchi and A.~Gaudillière.
\newblock Metastable states, quasi-stationnary distributions and soft measures.
\newblock {\em Stoch. proc. and appl.}, 126:1622--1680, 2016.

\bibitem{blackman.vigna:21}
David Blackman and Sebastiano Vigna.
\newblock Scrambled linear pseudorandom number generators.
\newblock {\em ACM Trans. Math. Softw.}, 47(4), sep 2021.

\bibitem{cassandro}
M.~Cassandro, A.~Galves, E.~Olivieri, and M.~E. Vares.
\newblock Metastable behavior of stochastic dynamics: A pathwise approach.
\newblock {\em Journal of Statistical Physics}, 35(5-6):603–634, Jun 1984.

\bibitem{champagnat}
N.~Champagnat and D.~Villemonais.
\newblock Quasi-stationary distributions in reducible state spaces.
\newblock arXiv:2201.1015.

\bibitem{compte.ea:00}
A.~Compte, N.~Brunel, P.~S. Goldman-Rakic, and X.~Wang.
\newblock {Synaptic Mechanisms and Network Dynamics Underlying Spatial Working
  Memory in a Cortical Network Model}.
\newblock {\em Cerebral Cortex}, 10(9):910--923, 09 2000.

\bibitem{constantinidis.ea:18}
C.~Constantinidis, S.~Funahashi, D.~Lee, J.~D. Murray, X.~Qi, M.~Wang, and
  A.~Arnsten.
\newblock Persistent spiking activity underlies working memory.
\newblock {\em The Journal of Neuroscience}, 38(32):7020–7028, Aug 2018.

\bibitem{darroch}
J.N. Darroch and E.~Seneta.
\newblock On quasi-stationnary distributions in absorbing continuous-time
  finite markov chains.
\newblock {\em Journal of Applied Probability}, 4:192--196, 1967.

\bibitem{defelice:81}
L.~J. DeFelice.
\newblock {\em Introduction to Membrane Noise}.
\newblock Plenum Press, 1981.

\bibitem{del_Castillo_1954}
J.~del Castillo and B.~Katz.
\newblock Quantal components of the end-plate potential.
\newblock {\em The Journal of Physiology}, 124(3):560–573, Jun 1954.

\bibitem{schonmann:1988}
R.~Durrett and R.~Schonmann.
\newblock The contact process on a finite set. ii.
\newblock {\em Annals Appl. Prob.}, 16:1570--1583, 1988.

\bibitem{Fatt_1952}
P.~Fatt and B.~Katz.
\newblock Spontaneous subthreshold activity at motor nerve endings.
\newblock {\em The Journal of Physiology}, 117(1):109--128, 1952.

\bibitem{feller:68}
William Feller.
\newblock {\em An Introduction to Probability Theory and its Applications.},
  volume~1.
\newblock John Wiley \& Sons, third edition, 1968.

\bibitem{fuster:73}
J.~Fuster.
\newblock Unit activity in prefrontal cortex during delayed-response
  performance: neuronal correlates of transient memory.
\newblock {\em Journal of Neurophysiology}, 36(1):61–78, Jan 1973.

\bibitem{fuster:15}
J.~Fuster.
\newblock {\em The Prefrontal Cortex}.
\newblock Elsevier Science, 2015.

\bibitem{galves.locherbach.ea:19}
A.~Galves, E.~L\"ocherbach, C.~Pouzat, and E.~Presutti.
\newblock A system of interacting neurons with short term synaptic
  facilitation.
\newblock {\em Journal of Statistical Physics}, 178(4):869–892, Dec 2019.

\bibitem{hansel.mato:13}
D.~Hansel and G.~Mato.
\newblock Short-term plasticity explains irregular persistent activity in
  working memory tasks.
\newblock {\em Journal of Neuroscience}, 33(1):133–149, Jan 2013.

\bibitem{huisinga:2004}
W.~Huisinga, S.~Meyn, and Christof Schütte.
\newblock Phase transitions and metastability in markovian and molecular
  systems.
\newblock {\em Annals Appl. Prob.}, 14:419--458, 2004.

\bibitem{iosifescu:07}
Marius Iosifescu.
\newblock {\em Finite Markov Processes and Their Applications}.
\newblock Dover books on mathematics. Dover Publications, 2007.

\bibitem{itskov.ea:11}
V.~Itskov, D.~Hansel, and M.~Tsodyks.
\newblock Short-term facilitation may stabilize parametric working memory
  trace.
\newblock {\em Frontiers in Computational Neuroscience}, 5:40, 2011.

\bibitem{katz.miledi:70}
B.~Katz and R.~Miledi.
\newblock Membrane noise produced by acetylcholine.
\newblock {\em Nature}, 226(5249):962–963, Jun 1970.

\bibitem{kijima:97}
M.~Kijima.
\newblock {\em Markov Processes for Stochastic Modeling}.
\newblock Springer New York, 1997.

\bibitem{leavitt.ea:17}
M.~L. Leavitt, D.~Mendoza-Halliday, and J.~C. Martinez-Trujillo.
\newblock Sustained activity encoding working memories: Not fully distributed.
\newblock {\em Trends in Neurosciences}, 40(6):328 -- 346, 2017.

\bibitem{liggettMarkovProcesses}
T.~Liggett.
\newblock {\em Continuous Time Markov Processes: An Introduction}.
\newblock American Mathematical Soc, 2010.

\bibitem{luo:15}
L.~Luo.
\newblock {\em Principles of Neurobiology}.
\newblock Garland Science, Jul 2015.

\bibitem{markus:12}
Arjen Markus.
\newblock {\em MODERN FORTRAN IN PRACTICE}.
\newblock Cambridge University Press, 2012.

\bibitem{meleard:2012}
S.~Meleard and D.~Villemonais.
\newblock Quasi-stationnary distributions and population processes.
\newblock {\em Probability Surveys}, 9:340--410, 2012.

\bibitem{mongillo.ea:08}
G.~Mongillo, O.~Barak, and M.~Tsodyks.
\newblock Synaptic theory of working memory.
\newblock {\em Science}, 319(5869):1543--1546, 2008.

\bibitem{mountford:1993}
T.~S. Mountford.
\newblock A metastable result for the finite multidimensional contact process.
\newblock {\em Canad. Math. Bul.}, 36:216--226, 1993.

\bibitem{rubinstein.kroese:16}
Reuven~Y. Rubinstein and Dirk~P. Kroese.
\newblock {\em Simulation and the Monte Carlo Method}.
\newblock Wiley Series in Probability and Statistics. Wiley, 2016.

\bibitem{schonmann}
R.H. Schonmann.
\newblock Metastability for the contact process.
\newblock {\em Journal of Statistical Physics}, 41:445--464, 1985.

\bibitem{verveen.derkesen:68}
A.A. Verveen and H.E. Derksen.
\newblock Fluctuation phenomena in nerve membrane.
\newblock {\em Proceedings of the IEEE}, 56(6):906–916, 1968.

\bibitem{wang:01}
X~J Wang.
\newblock Synaptic reverberation underlying mnemonic persistent activity.
\newblock {\em Trends Neurosci}, 24(8):455--63, Aug 2001.

\bibitem{wang.ea:06}
Y.~Wang, H.~Markram, P.~H. Goodman, T.~K Berger, J.~Ma, and P.~S.
  Goldman-Rakic.
\newblock Heterogeneity in the pyramidal network of the medial prefrontal
  cortex.
\newblock {\em Nature Neuroscience}, 9(4):534–542, Mar 2006.

\bibitem{Yarom_2011}
Y.~Yarom and J.~Hounsgaard.
\newblock Voltage fluctuations in neurons: Signal or noise?
\newblock {\em Physiological Reviews}, 91(3):917–929, Jul 2011.

\bibitem{zipser.kehoe.ea:93}
D.~Zipser, B.~Kehoe, G.~Littlewort, and J.~Fuster.
\newblock A spiking network model of short-term active memory.
\newblock {\em The Journal of Neuroscience}, 13(8):3406–3420, Aug 1993.

\end{thebibliography}

\end{document}